\documentclass[jmp,amsmath,amssymb,preprint,nobibnotes,secnumarabic]{revtex4-2}

\usepackage{multirow}
\usepackage{amsmath,amsfonts,amssymb,latexsym,amsthm}
\usepackage{color}
\usepackage{graphicx}
\usepackage{mathrsfs}
\usepackage{physics}
\usepackage{mathtools}
\usepackage{mathptmx}
\usepackage{dcolumn}
\usepackage{bm}
\usepackage{stmaryrd}
\usepackage[utf8]{inputenc}
\usepackage{fancyhdr}
\usepackage[T1]{fontenc}
\usepackage{soul}
\usepackage{tikz}
\usetikzlibrary{decorations.pathreplacing,calc,arrows.meta,positioning}
\usepackage{standalone}
\usepackage{tkz-base}
\usepackage{tikz-3dplot}
\usepackage{pgfplots}
\pgfplotsset{compat=1.15}
\usepgfplotslibrary{groupplots}
\usepackage{epsfig,epstopdf}
\usepackage[all]{xy}
\usepackage[unicode=false]{hyperref}
\usepackage{url}
\usepackage{dsfont}
\usepackage{slashed}
\usepackage{bbm}
\usepackage{enumitem}
\usepackage{silence}

\definecolor{myblue}{HTML}{1C3F6E}
\definecolor{myteal}{HTML}{0D7377}
\definecolor{myaccent}{HTML}{14BDAC}
\definecolor{mybg}{HTML}{F4F8FB}
\definecolor{hervecolor}{rgb}{0.8,0,0.7}

\newtheorem{prop}{Proposition}

\newenvironment{remark}[1][Remark]{\begin{trivlist}
\item[\hskip\labelsep{\bfseries #1}]}{\end{trivlist}}
\newenvironment{example}[1][Example]{\begin{trivlist}
\item[\hskip\labelsep{\bfseries #1}]}{\end{trivlist}}

\newcommand{\beprop}{\begin{prop}}
\newcommand{\enprop}{\end{prop}}
\newcommand{\bprf}{\begin{proof}}
\newcommand{\eprf}{\end{proof}\qed}
\newcommand{\CP}{\mathbb{C}\mathrm{P}}
\newcommand{\Liouv}{\mathcal{L}}
\DeclareMathOperator{\Vol}{Vol}
\DeclareMathOperator{\diag}{diag}

\newcommand{\SU}{\mathrm{SU}}
\newcommand{\scalar}[2]{\langle\kern.3ex #1\kern.3ex|\kern.3ex#2\kern.3ex\rangle}

\newcommand{\eu}{\mathsf{e}}
\newcommand{\ii}{\mathsf{i}}

\newcommand{\br}{\mathbf{r}}
\newcommand{\bs}{\mathbf{s}}

\newcommand{\bu}{\mathbf{u}}

\newcommand{\bv}{\mathbf{v}}
\def\be{\mathbf{e}}
\def\Pr{\mathsf{P}}
\def\ba{\mathbf{a}}
\def\un{\mathbbm{1}}

\newcommand{\R}{\mathbb{R}}

\newcommand{\C}{\mathbb{C}}

\def\lg{\langle}
\newcommand{\rg}{\rangle}

\newcommand{\ud}{\mathrm{d}}

\newcommand{\Hc}[1]{H_{#1}}
\newcommand{\E}[2]{E_{#1#2}}

\def\vph{\pmb{\phi}}

\def\sfD{\mathsf{D}}

\begin{document}
\numberwithin{equation}{section}

\title{Spectral-angular parametrization of open qudit dynamics:\\
gap coordinates, Cartan structure, and GKLS decoupling}

\author{Jean-Pierre Gazeau}
\email{gazeau@apc.in2p3.fr, j.gazeau@uwb.edu.pl}
\affiliation{Universit\'e Paris Cit\'{e}, CNRS, Astroparticule et Cosmologie,
  F-75013 Paris, France}
\affiliation{Faculty of Mathematics, University of Bia\l ystok,
  15-245 Bia\l ystok, Poland}

\author{Kaoutar El Bachiri}
\email{kaoutar-elbachiri@um5r.ac.ma}
\affiliation{ESMaR, Faculty of Sciences at Univ. Mohammed V - Rabat, 1014 Rabat, Morocco}

\author{Zakaria Bouameur}
\email{zakaria.bouameur@um5r.ac.ma}
\affiliation{ESMaR, Faculty of Sciences at Univ. Mohammed V - Rabat, 1014 Rabat, Morocco}

\author{Yassine Hassouni}
\email{y.hassouni@um5r.ac.ma}
\affiliation{ESMaR, Faculty of Sciences at Univ. Mohammed V - Rabat, 1014 Rabat, Morocco}

\date{\today}

\begin{abstract}

Every $n$-level density matrix can be diagonalized as
$\rho=U\diag(p_1,\ldots,p_n)U^\dagger$, splitting it into an
eigenvalue vector $(p_1,\ldots,p_n)$ and a unitary eigenbasis $U$.
We introduce a new coordinate system on the eigenvalue vector alone:
the \emph{gap coordinates} $r_a:=p_a-p_{a+1}\ge0$, $a=1,\ldots,n-1$.
 We show that
$(r_1,\ldots,r_{n-1})$ are precisely the simple-root coordinates of the
eigenvalue vector in the Cartan subalgebra of $A_{n-1}=\mathfrak{sl}(n)$, so
that the spectral diagonal expands in the fundamental-coweight basis,
$\mathsf D(\mathbf r)=\sum_a r_a\,\omega_a^\vee$, and the $n-1$
eigenvalue-ordering inequalities collapse into the single linear constraint
$\sum_a a\,r_a\le1$ defining a weighted simplex $R_{n-1}$. This
identification yields closed-form expressions for the Fisher--Rao and Bures
metrics at the maximally mixed state in terms of the inverse Cartan matrix
of $A_{n-1}$, and for a piecewise-linear alternative to the  standard definition of purity.
On the angular side we
give an explicit Tilma--Sudarshan-type coordinatization of the flag manifold
$\mathcal F_n\simeq\mathrm{SU}(n)/\mathbb T^{n-1}$ and a global coherent-state
resolution of the identity on $\mathcal F_n$, enabling an
$\mathrm{SU}(n)$-covariant integral quantization of the flag variety. Under
GKLS dynamics, the pair (gap coordinates, flag angles) exhibits a partial
decoupling: the eigenvalue gaps evolve under dissipation alone, while the
flag angles are driven by both the Hamiltonian and the dissipator, which we
derive directly in these coordinates, including an explicit illustration on
the real qutrit sector. 
\end{abstract}

\maketitle

\noindent\textbf{MSC classes} 81P16, 81R05, 81R30, 81S22


\tableofcontents

\section{Introduction}
\label{intro}

Finite-dimensional (mixed) quantum states, represented by density matrices,
occupy a central place in the formalism of quantum mechanics and quantum optics.
Their geometric and metric content can reveal non-trivial aspects pertaining to
``quantum geometry''~\cite{ZS2003,BZ2006}.  These aspects develop at an
intricate level as soon as one deals with tensor products of states, with
accompanying physical concepts such as entanglement, fidelity, purity, and
quantum correlations.

From a mathematical viewpoint, the space $\mathcal{D}_n$ of density matrices on a
finite-dimensional Hilbert space $\mathbb{C}^n$ forms a compact convex body
embedded in the real vector space of Hermitian matrices with unit trace.  Its
extreme points correspond to pure states, while mixed states lie in the
interior.  This structure motivates the study of various parametrizations,
either linear (e.g.\ generalized Bloch vectors) or nonlinear (spectral
decompositions, group-theoretical constructions), which are useful for quantum
state reconstruction, numerical simulations, and the analysis of dynamical
processes.

When the quantum system interacts with an environment, its evolution is no
longer unitary and must be described within the framework of open quantum
systems.  In the Markovian regime the dynamics is governed by the
Gorini--Kossakowski--Lindblad--Sudarshan (GKLS)
equation~\cite{breupetr02,alilend07,rivas12,chruscinski-pascazio17}.  Although
this equation provides a complete characterization of completely positive
trace-preserving dynamical semigroups, its explicit analysis becomes cumbersome
as the system dimension increases, because the density matrix simultaneously
carries spectral information (eigenvalues) and geometric information (eigenvectors).
To illustrate this point, we recall the basic qubit case~\cite{macugako26}.

\subsubsection*{Qubits}

Figure~\ref{bloch} shows the unit Bloch ball in $\R^3$, with boundary the Bloch
sphere.  To a point $A$ in the ball, i.e.\ to the vector
$$
\overrightarrow{OA}\equiv\ba=r\,\mathbf{u}(\theta,\phi)
=(r\sin\theta\cos\phi,\,r\sin\theta\sin\phi,\,r\cos\theta),
\quad 0\le r\le1,
$$
there corresponds the $2\times2$ density matrix with its spectral decomposition
\begin{align}
\label{densmat11}
\rho_{r,\theta,\phi}
&=\frac{1+r}{2}|\mathbf{u}_{1}\rg\lg\mathbf{u}_{1}|
 +\frac{1-r}{2}|\mathbf{u}_{2}\rg\lg\mathbf{u}_{2}|
 =p_1|\mathbf{u}_{1}\rg\lg\mathbf{u}_{1}|+p_2|\mathbf{u}_{2}\rg\lg\mathbf{u}_{2}|,\\
\label{densmat12}
&=\frac{1}{2}\!\left(\un_2+r\,U(\theta,\phi)\sigma_3 U^{\dag}(\theta,\phi)\right),\\
\label{densmat13}
&=\tfrac12\!\left(\un_2+{\mathbf a}\cdot\boldsymbol\sigma\right),
\end{align}
with $p_1,p_2 \in [0,1]$, $p_1 + p_2=1$, 
\begin{equation}
\label{Un2}
U(\theta,\phi):=\begin{pmatrix}
  \cos\frac{\theta}{2} & -\sin\frac{\theta}{2}\eu^{-\ii\phi}\\
  \sin\frac{\theta}{2}\eu^{\ii\phi} & \cos\frac{\theta}{2}
\end{pmatrix},\quad\theta\in[0,\pi],\quad\phi\in[0,2\pi),
\end{equation}
and Pauli matrices
\begin{equation}
\label{pauli}
\boldsymbol\sigma=\left(\sigma_1=\begin{pmatrix}0&1\\1&0\end{pmatrix},\;
\sigma_2=\begin{pmatrix}0&-\ii\\\ii&0\end{pmatrix},\;
\sigma_3=\begin{pmatrix}1&0\\0&-1\end{pmatrix}\right).
\end{equation}
One checks that $U(\theta,\phi)=\eu^{-\ii\frac{\phi}{2}\sigma_3}
\eu^{-\ii\frac{\theta}{2}\sigma_2}\eu^{\ii\frac{\phi}{2}\sigma_3}$,
which is an element of the left coset $\mathrm{SU}(2)/\mathrm{U}(1)$.
Its column vectors
$$
\bu_1=\begin{pmatrix}\cos\frac{\theta}{2}\\\sin\frac{\theta}{2}\eu^{\ii\phi}\end{pmatrix}\equiv |\theta,\phi\rg,
\quad
\bu_2=\begin{pmatrix}-\sin\frac{\theta}{2}\eu^{-\ii\phi}\\\cos\frac{\theta}{2}\end{pmatrix} \equiv \eu^{-\ii\phi}|\pi-\theta,\phi\rg,
$$
form an orthonormal basis of $\C^2$ and can be viewed as spin-$\tfrac12$ coherent
states.  They satisfy the resolution of the identity
\begin{equation}
\label{resid2}
\int_{\mathbb{S}^2}\frac{\ud\bu}{2\pi}\,|\theta,\phi\rg\lg\theta,\phi|=\int_{\mathbb{S}^2}\frac{\ud\bu}{2\pi}\,|\pi-\theta,\phi\rg\lg\pi -\theta,\phi|=\un_2,
\quad\ud\bu=\sin\theta\,\ud\theta\,\ud\phi,
\end{equation}
and consequently,
\begin{equation}
\label{residrho}
\int_{\mathbb{S}^2}\rho_{r,\theta,\phi}\frac{\ud\bu}{2\pi}\,=\un_2.
\end{equation}
The Bloch radius $r=2p_1-1\in[0,1]$ is the normalized excess of the maximal
eigenvalue $p_1$ above the uniform baseline; $r=1$ gives pure states and $r=0$
gives the totally random state $\tfrac12\un_2$.

\begin{figure}[tbp]
\begin{center}
\tdplotsetmaincoords{62}{120}
\begin{tikzpicture}[tdplot_main_coords,scale=4]
  \def\R{1}
  \draw[->,thick](0,0,0)--(1.3,0,0) node[below left]{$x$};
  \draw[->,thick](0,0,0)--(0,1.3,0) node[right]{$y$};
  \draw[->,thick](0,0,0)--(0,0,1.3) node[above]{$z$};
  \begin{scope}[tdplot_screen_coords]
    \shade[ball color=blue!25,opacity=0.35](0,0) circle(\R);
    \draw[blue!60](0,0) circle(\R);
    \draw[densely dotted,gray!60](0,0) ellipse[x radius=\R,y radius=0.35*\R];
  \end{scope}
  \def\rpure{1.0}\def\thetap{35}\def\phip{20}
  \def\rmix{0.55}\def\thetam{60}\def\phim{120}
  \tdplotsetcoord{P}{\rpure}{\thetap}{\phip}
  \tdplotsetcoord{M}{\rmix}{\thetam}{\phim}
  \draw[very thin,gray!70,samples=90,domain=0:180,variable=\t]
    plot({\R*sin(\t)*cos(\phip)},{\R*sin(\t)*sin(\phip)},{\R*cos(\t)});
  \draw[very thin,gray!70,samples=90,domain=0:180,variable=\t]
    plot({\R*sin(\t)*cos(\phim)},{\R*sin(\t)*sin(\phim)},{\R*cos(\t)});
  \draw[very thick,-stealth,blue!70!black](0,0,0)--(P);
  \fill[blue!80!black](P) circle(0.7pt);
  \draw[very thick,-stealth,red!80!black](0,0,0)--(M);
  \fill[red!80!black](M) circle(0.7pt);
  \tdplotdrawarc[->,red!60,densely dashed]{(0,0,0)}{0.45}{0}{\phim}{anchor=north}{$\phi$}
  \draw[black!60,densely dashed,->,samples=40,domain=0:\thetam,variable=\t]
    plot({0.45*sin(\t)*cos(\phim)},{0.45*sin(\t)*sin(\phim)},{0.45*cos(\t)});
  \draw[densely dashed,gray!70](M)--(Mxy);
  \pgfmathsetmacro{\Mz}{\rmix*cos(\thetam)}
  \draw[densely dashed,gray!70](0,0,\Mz)--(M);
  \node[blue!80!black,anchor=south east,yshift=-4pt] at (P){pure $(r=1)$};
  \node[blue!80!black,anchor=south east,yshift=-15pt] at (P){$|\theta',\phi'\rg\lg\theta',\phi'|$};
  \node[red!80!black,anchor=north west,xshift=-10pt,yshift=-4pt] at (M){mixed $(r<1)$};
  \node[red!80!black,anchor=north west,yshift=10pt] at (M){$\rho_{r,\theta,\phi}$};
  \node[red!80!black,anchor=north west,xshift=-25pt,yshift=25pt] at (M){$\theta$};
\end{tikzpicture}
\caption{To each point $A$ in the unit closed (Bloch) ball corresponds the
  quantum state $\rho_{r,\theta,\phi}$.}
\label{bloch}
\end{center}
\end{figure}
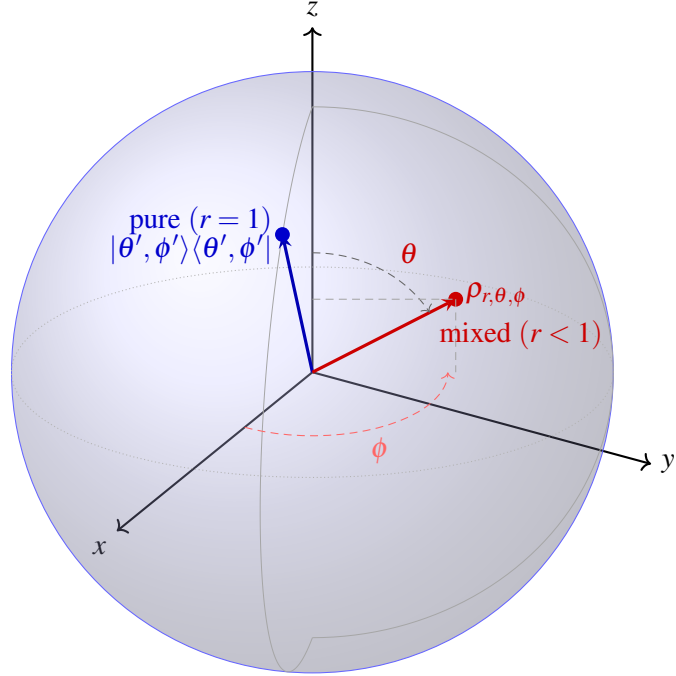

Under unitary Hamiltonian action coupled with Markovian environment interaction,
the qubit obeys the GKLS master equation ($\hbar=1$):
\begin{equation}
\label{lindblad2}
\frac{\ud\rho}{\ud t}
=-\ii[H,\rho]+\sum_{k}h_k\!\left(L_k\rho L_k^\dagger
-\tfrac12\{\rho,L_k^\dagger L_k\}\right)
\equiv\Liouv(\rho)=\Liouv_{\rm un}(\rho)+\Liouv_{\rm diss}(\rho).
\end{equation}
Choosing $L_k=\sigma_k$ ($k=x,y,z\equiv1,2,3$) and
$H=\begin{pmatrix}h_{00}&h_{01}\\\overline{h}_{01}&h_{11}\end{pmatrix}$,
the GKLS equation yields the following equations for the angular parts,
\begin{align}
\label{dphi}
\dot\phi
&=h_{00}-h_{11}-2\cot\theta\big(\Re h_{01}\cos\phi-\Im h_{01}\sin\phi\big)
+(h_2-h_1)\sin2\phi,\\
\label{dthet}
\dot\theta
&=-2(\Re h_{01}\sin\phi+\Im h_{01}\cos\phi)
+\sin2\theta\,(h_1\cos^2\phi+h_2\sin^2\phi-h_3),
\end{align}
and for the radial (spectral) part,
\begin{equation}
\label{dr}
\frac{\dot r}{r}
=-2\Big[h_1\big(1-\sin^2\theta\cos^2\phi\big)
+h_2\big(1-\sin^2\theta\sin^2\phi\big)
+h_3\sin^2\theta\Big]\le0.
\end{equation}
Equation~\eqref{dr} shows that the Bloch radius $r$ (a measure of purity) is
unaffected by the Hamiltonian and decreases under dissipation, with the exception of stationary points discussed in \cite{macugako26}.

\subsubsection*{Toward generalization}

The aim of this work is to generalize the qubit picture to any finite
dimension, emphasizing the separation between parameters describing the spectrum
of a state $\rho\in \mathcal D_n$ and those describing its orientation in Hilbert space.  The occasionally elementary character of the presentation is intentional, reflecting our aim to combine a self-contained exposition with a pedagogical introduction to the main ideas.

Parametrizing $\mathcal D_n$
efficiently --- separating the ``shape'' of a state (its eigenvalues) from
its ``orientation'' (its eigenbasis) --- is a classical problem with a long
history: Bloch-vector expansions in an $\mathfrak{su}(n)$ generator
basis~\cite{BZ2006}, generalized Euler-angle parametrizations of
$\mathrm{SU}(N)$~\cite{TilmaSudarshan2002,BertiniCacciatoriCerchiai2006},
coset parametrizations~\cite{akhtarshenas07,akhtarshenas07bures}, and the
recursive Jarlskog parametrization~\cite{jarlskog05,bruning-etal11} all
instantiate, in one way or another, the spectral decomposition
\begin{equation}
\label{eq:generic-split}
\rho = U\,\diag(p_1,\ldots,p_n)\,U^\dagger,
\qquad
p_1\ge\cdots\ge p_n\ge0,\quad \textstyle\sum_k p_k=1,
\end{equation}
with $U$ ranging over a coset of $\mathrm{U}(n)$ that leaves the diagonal
matrix stable. When the system is open and Markovian, governed by the
Gorini--Kossakowski--Lindblad--Sudarshan (GKLS)
equation~\cite{breupetr02,alilend07,rivas12,chruscinski-pascazio17}, this
same splitting underlies the separation of the dynamics into an
\emph{inter-orbit} flow of the eigenvalue vector and an \emph{intra-orbit}
flow along the corresponding flag manifold, a viewpoint developed in detail
by Rooney, Bloch and Rangan for Lindblad control
systems~\cite{rooney-bloch-rangan16} and for the asymptotic/stationary
structure of such
systems~\cite{rooney-bloch-rangan18}.

None of this is in dispute. What this paper contributes is a specific
coordinate choice on the eigenvalue factor of~\eqref{eq:generic-split} ---
the \emph{gap coordinates} $r_a=p_a-p_{a+1}$ --- together with the
observation that this choice is not an arbitrary reparametrization of the
ordered simplex but coincides with the simple-root (fundamental-coweight)
coordinates of $A_{n-1}=\mathfrak{sl}(n)$. This single fact drives every
subsequent result: a one-inequality spectral domain, closed-form
Cartan-matrix expressions for the Fisher--Rao and Bures metrics near the
maximally mixed state, a piecewise-linear purity functional, and an explicit
decoupling of GKLS dynamics into a closed dissipation-only flow for $\mathbf
r$ and a Hamiltonian-and-dissipation-driven flow on the flag manifold. 
\medskip

The paper is organized as follows. 
Section~\ref{sec:statespace} introduces the space of density matrices 
$\rho_{\br,\vph}$ on $\mathbb{C}^n$ and presents the geometric structure 
of the spectral parametrization $\br$ in terms of a convex polytope 
$R_{n-1}\subset\mathbb{R}^{n-1}$; the cases $n=3$ and $n=4$ are worked 
out explicitly to illustrate this framework. 
Section~\ref{statcomp} examines the statistical properties of the spectral 
parameters $\br$ and contrasts them with those of the standard probability 
parameters. 
Section \ref{discint} is an intermediate discussion about the advantages of our $\br$ parametrisation with regard to the ordered probability simplex forming the set of eigenvalues of $\rho_{\br,\vph}$. 
Section~\ref{sec:purity} proposes an alternative notion of purity formulated 
directly in terms of $\br$. 
Section~\ref{SUPar} develops the angular parametrization $\vph$ of $\rho_{\br,\vph}$
via the flag manifold 
$\mathcal{F}_n=\mathrm{SU}(n)/\mathbb{T}^{n-1}$,  with particular emphasis on the Perelomov coherent state 
interpretation, the associated resolution of the identity, and the resulting 
$\mathrm{SU}(n)$-covariant integral quantization. The case $n=3$ is treated 
in detail. 
Section~\ref{sec:gkls} applies the full parametrization to GKLS dynamics 
and demonstrates how the evolution equations decouple into independent 
spectral and angular sectors. A summary of the real case $n=3$  is given as an illustration. 
We
state the precise relation of each of our results to the  papers \cite{akhtarshenas07,akhtarshenas07bures,jarlskog05,bruning-etal11} cited
above in Section~\ref{sec:related}, by insisting that  we do not repeat, as new claims,
material that already appears there.

Concluding remarks and perspectives are collected in 
Section~\ref{sec:conclusion}.
\section{Pure and mixed quantum states in \texorpdfstring{$n$}{n}-dimensional Hilbert space}
\label{sec:statespace}

A vector in $\C^n$ is denoted $\bv$ (Euclidean/Hermitian geometry) or $|\bv\rg$
(Dirac notation).  A mixed state is represented by a density matrix
\begin{equation}
\label{nrho}
\rho\equiv\rho_{\br,\vph}
=\sum_{i=1}^n p_i\,\Pr_{\bu_i}
\equiv\frac{1}{n}\un_n+U(\vph)\,\sfD(\br)\,U^{\dag}(\vph)\equiv U(\vph)\,\rho_\br\,U^{\dag}(\vph),
\end{equation}
where $0\le p_i\le1$ with $\sum_{i=1}^n p_i=1$ are the eigenvalues of $\rho$,
and $\Pr_{\bu_i}=|\bu_i\rg\lg\bu_i|$ are the corresponding rank-one projectors.
The diagonal matrix $\frac{1}{n}\,\un_n \equiv \rho_{\mathrm{rm}}$ represents the
maximally mixed state. The diagonal matrix $\sfD(\br)$ depends linearly on an $(n-1)$-tuple
$\br = (r_1,r_2,\dotsc,r_{n-1})$, whose definition will be specified below. 

If $\rho_{\br,\vph}$ is non-degenerate (simple spectrum), the number of
distinct real parameters is exactly $n^2-1$.  The set of degenerate density
matrices has measure zero in the full state space.

\subsection{Spectral parametrization \texorpdfstring{$\br$}{r} in gap coordinates}
\label{sec:param}


Given an ordered probability distribution
\[
p_1\ge p_2\ge\cdots\ge p_n\ge0,
\qquad
\sum_{i=1}^n p_i=1,
\]
we define the \emph{spectral gap coordinates} by
\begin{equation}
\label{def:r}
r_a := p_a-p_{a+1}\ge0,
\qquad a=1,\ldots,n-1.
\end{equation}
These variables represent the successive gaps between adjacent eigenvalues.
For \(n=2\), one has \(r_1=p_1-p_2\), which coincides with the Bloch-ball radius; \eqref{def:r} may therefore be viewed as its natural \(n\)-level generalization.

\paragraph{Inverse relations.}
The eigenvalues are recovered from \(\br=(r_1,\ldots,r_{n-1})\) through
\begin{equation}
\label{eq:p_from_r}
p_k
=
\frac{1}{n}
+\sum_{a=k}^{n-1}\left(1-\frac{a}{n}\right)r_a
-\sum_{a=1}^{k-1}\frac{a}{n}\,r_a,
\qquad k=1,\ldots,n,
\end{equation}
with the convention that empty sums vanish. Equivalently,
\begin{equation}
\label{eq:p_from_r_compact}
p_k-\frac{1}{n}
=
\sum_{a=1}^{n-1} M_{ka}\,r_a,
\qquad
M_{ka}:=
\begin{cases}
1-\dfrac{a}{n}, & k\le a, \\[6pt]
-\dfrac{a}{n}, & k>a.
\end{cases}
\end{equation}
For each fixed \(a\), one has \(\sum_{k=1}^n M_{ka}=0\), consistently with the normalization condition \(\sum_{k=1}^n p_k=1\).
\paragraph{Eigenvalue polytope.}
The non-negativity constraints $p_k\ge0$ are automatic for
$k=1,\ldots,n-1$ once $r_a\ge0$, since $p_k\ge p_{k+1}\ge\cdots\ge p_n$.
The sole remaining constraint is $p_n\ge0$:
$$
p_n = \frac{1}{n}\!\left(1-\sum_{a=1}^{n-1}a\,r_a\right)\ge0.
$$
Hence the admissible parameter domain is the \emph{weighted simplex}
\begin{equation}
\label{def:polR}
R_{n-1} :=
\left\{\br\in\mathbb{R}^{n-1} :
r_a\ge0\ \forall a, \quad
\sum_{a=1}^{n-1}a\,r_a\le 1
\right\},
\end{equation}
whose vertices are $\br=\mathbf{0}$ (maximally mixed state) and
$\br=\mathbf{e}_a/a$ for $a=1,\ldots,n-1$ (boundary states with $p_1=\cdots=p_a=\frac{1}{a}$ and 
$p_{a+1}=\cdots=p_n=0$), and where $(\mathbf{e}_a)$ denotes  the canonical orthonormal basis of $\R^{n-1}$. The pure-state condition
$p_1=1$, $p_2=\cdots=p_n=0$ corresponds to $r_a = \delta_{1a}$.


\subsection{Relation to the spectral simplex and volume comparisons}

The standard $(n-1)$-dimensional probability simplex is
\[
\Delta_{n-1}:=\{(p_1,\ldots,p_{n-1})\in\R^{n-1}:p_i\ge0,\,\textstyle\sum_{i=1}^{n-1}p_i\le1\},
\]
with $p_n:=1-\sum_{i=1}^{n-1}p_i$.  The ordered simplex (Weyl chamber) is
\begin{equation}
\label{eq:ordered-simplex}
\Delta_{n-1}^{\downarrow}
:=\{\mathbf{p}\in\Delta_{n-1}:p_1\ge p_2\ge\cdots\ge p_{n-1}\ge p_n\ge0\}.
\end{equation}
Its volume is given by  
\begin{prop}
\begin{equation}
\label{Deln1Vol}
\Vol_{n-1}(\Delta_{n-1}^{\downarrow})=\frac{1}{n!\,(n-1)!}.
\end{equation}
\end{prop}
\begin{proof}
Setting $p_n:=1-\sum_{i=1}^{n-1}p_i$, the region $\Delta_{n-1}^{\downarrow}$
is the subset of the standard simplex $\Delta_{n-1}$ (which has volume $1/(n-1)!$)
for which $p_1\ge p_2\ge\cdots\ge p_n\ge0$.  By symmetry, the simplex is
partitioned (up to measure zero) into $n!$ congruent ordered regions, so
$\Vol_{n-1}(\Delta_{n-1}^{\downarrow})=\frac{1}{n!}\cdot\frac{1}{(n-1)!}$.
\end{proof}
\begin{prop}
The gap map $\Psi\colon(p_1,\ldots,p_n)\mapsto(r_1,\ldots,r_{n-1})$,
$r_a=p_a-p_{a+1}$, is a bijection between
$\Delta_{n-1}^{\downarrow}$ and the weighted simplex $R_{n-1}$
defined in~\eqref{def:polR}.
\end{prop}
\begin{proof}
The map $\Psi$ is linear with inverse given by~\eqref{eq:p_from_r}.
The ordering $p_1\ge\cdots\ge p_n$ translates to $r_a\ge0$, and the
normalisation $\sum_k p_k=1$ together with $p_n\ge0$ yields
$\sum_{a=1}^{n-1}a\,r_a\le1$; hence $\Psi(\Delta_{n-1}^{\downarrow})=R_{n-1}$.
\end{proof}

The Euclidean volume of $R_{n-1}$ is:
\begin{equation}
\label{Rn1Vol}
\Vol_{n-1}(R_{n-1})
=\frac{1}{\bigl((n-1)!\bigr)^2},
\end{equation}
and so the ratio of volumes weighted/ordered simplexes is
\begin{equation}
\label{compVolR}
\frac{\Vol_{n-1}(R_{n-1})}{\Vol_{n-1}(\Delta_{n-1}^{\downarrow})}
=n,
\end{equation}
\begin{proof}[Proof of~\eqref{Rn1Vol}]
The set $R_{n-1}$ is a simplex with $n$ vertices $\mathbf{0}$ and
$\mathbf{e}_a/a$ for $a=1,\ldots,n-1$.  Taking $\mathbf{0}$ as base
point, the edge matrix is
$M=\operatorname{diag}\left(1,\tfrac{1}{2},\ldots,\tfrac{1}{n-1}\right)$,
so
$$
\Vol_{n-1}(R_{n-1})
=\frac{|\det M|}{(n-1)!}
=\frac{1}{(n-1)!}\cdot\frac{1}{(n-1)!}
=\frac{1}{\bigl((n-1)!\bigr)^2}.
$$
Alternatively, $\Psi$ has constant Jacobian $|\det\mathrm{d}\Psi|=n$
on the hyperplane $\sum_k p_k=1$ (verified by direct expansion for any $n$),
which gives $\Vol(R_{n-1})=n\cdot\Vol(\Delta_{n-1}^{\downarrow})
=n/(n!\,(n-1)!)=1/((n-1)!)^2$.
\end{proof}

\begin{remark}[Geometric meaning of $R_{n-1}$]
The weighted simplex $R_{n-1}$ is the image of the ordered probability
simplex under the \emph{linear} gap map $\Psi$.  Its vertices
$\mathbf{0}$ and $\mathbf{e}_a/a$ correspond, respectively, to the
maximally mixed state ($p_k=1/n$ for all $k$) and to the boundary
states with $p_1=\cdots=p_a=1/a$, $p_{a+1}=\cdots=p_n=0$.
The pure state $p_1=1$, $p_k=0$ $(k>1)$ is the vertex $\mathbf{e}_1=(1,0,\ldots,0)$,
the farthest from the origin.  Since all edges of $R_{n-1}$ originate
from $\mathbf{0}$, the polytope is a \emph{coordinate simplex} stretched
anisotropically: edge $a$ has length $1/a$, so higher-index gaps are
geometrically compressed.
\end{remark}

\subsection{Geometric illustrations for \texorpdfstring{$n=3$}{n=3}
  and \texorpdfstring{$n=4$}{n=4}}

\subsubsection*{Simplex \texorpdfstring{$R_2$}{R2} for \texorpdfstring{$n=3$}{n=3}}

The constraints~\eqref{def:polR} reduce to
$$
r_1\ge0,\quad r_2\ge0,\quad r_1+2r_2\le1.
$$
The polytope $R_2$ is a triangle with vertices $(0,0)$, $(1,0)$,
$(0,\tfrac12)$.
Its area is $\Vol_2(R_2)=\tfrac14$, and the ratio $\Vol(R_2)/\Vol(\Delta_2^{\downarrow})=3$.
For comparison, it is shown  in Fig.~\ref{fig:R2Delta2} together with the ordered probability simplex $\Delta_2^{\downarrow}$.

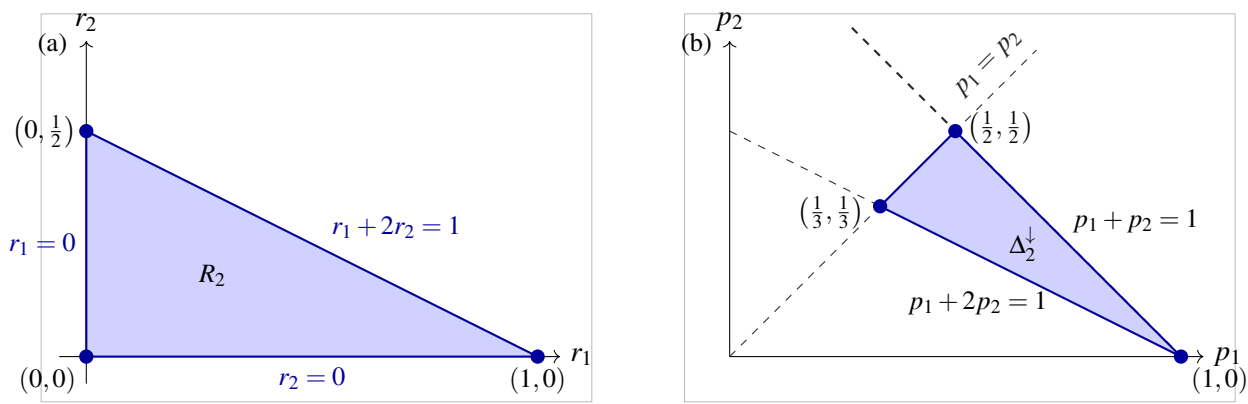
\begin{figure}[htb!]
\centering

\begin{minipage}[t]{0.48\textwidth}
\centering
\begin{tikzpicture}[scale=6]
  \path[use as bounding box] (-0.12,-0.12) rectangle (1.16,0.80);

  \draw[black!25] (-0.10,-0.10) rectangle (1.12,0.76);

  \node[anchor=north west,font=\footnotesize] at (-0.13,0.74) {(a)};

  \draw[->] (-0.06,0) -- (1.05,0) node[right] {\small $r_1$};
  \draw[->] (0,-0.06) -- (0,0.70) node[above] {\small $r_2$};

  \fill[blue!18] (0,0) -- (1,0) -- (0,0.5) -- cycle;

  \draw[thick,blue!60!black] (0,0) -- (1,0)
    node[pos=0.5,below,font=\footnotesize] {$r_2=0$};
  \draw[thick,blue!60!black] (0,0) -- (0,0.5)
    node[pos=0.5,left,font=\footnotesize] {$r_1=0$};
  \draw[thick,blue!60!black] (1,0) -- (0,0.5)
    node[pos=0.48,above right,font=\footnotesize] {$r_1+2r_2=1$};

  \foreach \x/\y in {0/0,1/0,0/0.5}
    \fill[blue!60!black] (\x,\y) circle (0.45pt);

  \node[font=\footnotesize,below left] at (0,0) {$(0,0)$};
  \node[font=\footnotesize,below]      at (1,0) {$(1,0)$};
  \node[font=\footnotesize,left]       at (0,0.5) {$\bigl(0,\tfrac12\bigr)$};

  \node[font=\footnotesize] at (0.28,0.18) {$R_2$};
\end{tikzpicture}
\end{minipage}
\hfill
\begin{minipage}[t]{0.48\textwidth}
\centering
\begin{tikzpicture}[scale=6]
  \path[use as bounding box] (-0.12,-0.12) rectangle (1.16,0.80);

  \draw[black!25] (-0.10,-0.10) rectangle (1.12,0.76);

  \node[anchor=north west,font=\footnotesize] at (-0.13,0.74) {(b)};

  \draw[->] (0,0) -- (1.05,0) node[right] {\small $p_1$};
  \draw[->] (0,0) -- (0,0.70) node[above] {\small $p_2$};

  \draw[black!85,dashed] (0,0) -- (0.68,0.68)
node[pos=0.90,sloped,anchor=south,yshift=2pt,font=\footnotesize] {$p_1=p_2$};
 
\draw[black!85,dashed] (0,0.5) -- (1,0);
\node[font=\footnotesize] at (0.55,0.12) {$p_1+2p_2=1$};
  
  \draw[black!85,dashed,thick,shorten <=5pt] (1/4,3/4) -- (1,0);
\node[font=\footnotesize] at (0.9,0.3) {$p_1+p_2=1$};
  
  \fill[blue!18] (1,0) -- (0.5,0.5) -- (1/3,1/3) -- cycle;
  \draw[thick,blue!60!black] (1,0) -- (0.5,0.5) -- (1/3,1/3) -- cycle;

  \fill[blue!60!black] (1,0)     circle (0.45pt);
  \fill[blue!60!black] (0.5,0.5) circle (0.45pt);
  \fill[blue!60!black] (1/3,1/3) circle (0.45pt);

  \node[font=\footnotesize,below right] at (1,0) {$(1,0)$};
  \node[font=\footnotesize,above]       at (0.6,0.45) {$\bigl(\tfrac12,\tfrac12\bigr)$};
  \node[font=\footnotesize,left]        at (0.32,0.32) {$\bigl(\tfrac13,\tfrac13\bigr)$};

  \node[font=\footnotesize] at (0.65,0.25) {$\Delta_2^{\downarrow}$};
\end{tikzpicture}
\end{minipage}

\caption{Comparison for $n=3$: (a) the weighted simplex $R_2$ in $(r_1,r_2)$ and (b) the ordered probability simplex $\Delta_2^{\downarrow}$ in $(p_1,p_2)$. The two domains are related by the linear gap map.}
\label{fig:R2Delta2}
\end{figure}

\subsubsection*{Simplex \texorpdfstring{$R_3$}{R3} for \texorpdfstring{$n=4$}{n=4}}

The constraints reduce to
$$
r_1\ge0,\quad r_2\ge0,\quad r_3\ge0,\quad r_1+2r_2+3r_3\le1.
$$
The polytope $R_3\subset\mathbb{R}^3$ is a \emph{tetrahedron} with four vertices
$$
V_1=(0,0,0),\quad
V_2=(1,0,0),\quad
V_3=\bigl(0,\tfrac{1}{2},0\bigr),\quad
V_4=\bigl(0,0,\tfrac{1}{3}\bigr),
$$
four triangular facets
$$
F_1\colon r_1=0,\quad
F_2\colon r_2=0,\quad
F_3\colon r_3=0,\quad
F_4\colon r_1+2r_2+3r_3=1,
$$
and volume $\Vol_3(R_3)=\tfrac{1}{36}$.
The tetrahedron is shown in Fig.~\ref{fig:R3}.


\begin{figure}[tbp!]
\centering
\begin{tikzpicture}[
  x={(-2.8cm,-0.80cm)},
  y={( 4.2cm,-0.80cm)},
  z={( 0.0cm, 4.2cm)},
  vertex/.style={circle, fill=black, inner sep=1.1pt},
  tick/.style={font=\footnotesize}
]

\draw[->, thick] (0,0,0) -- (1.2,0,0);
\draw[->, thick] (0,0,0) -- (0,0.7,0);
\draw[->, thick] (0,0,0) -- (0,0,0.6);

\node[tick, anchor=north] at (0,0,0) {$0$};
\node[tick, anchor=north] at (1,0,0) {$1$};
\node[tick, anchor=north east] at (0,0.6,0) {$\frac12$};
\node[tick, anchor=west] at (0,0,0.4) {$\frac13$};

\coordinate (V1) at (0,0,0);
\coordinate (V2) at (1,0,0);
\coordinate (V3) at (0,0.5,0);
\coordinate (V4) at (0,0,1/3);

\colorlet{facecolor}{cyan!55!blue}

\fill[facecolor, fill opacity=0.25, draw=black!45, line width=0.4pt]
  (V1)--(V3)--(V4)--cycle;

\fill[facecolor, fill opacity=0.30, draw=black!45, line width=0.4pt]
  (V1)--(V2)--(V3)--cycle;

\fill[facecolor, fill opacity=0.30, draw=black!45, line width=0.4pt]
  (V1)--(V2)--(V4)--cycle;

\fill[facecolor, fill opacity=0.55, draw=black!65, line width=0.5pt]
  (V2)--(V3)--(V4)--cycle;

\draw[black, line width=1pt]
  (V1)--(V2) (V1)--(V3) (V1)--(V4)
  (V2)--(V3) (V2)--(V4) (V3)--(V4);

\node[vertex] at (V1) {};
\node[vertex] at (V2) {};
\node[vertex] at (V3) {};
\node[vertex] at (V4) {};

\node[font=\normalsize, anchor=north] at (1.25,0,0) {$r_1$};
\node[font=\normalsize, anchor=north west] at (0,0.7,0) {$r_2$};
\node[font=\normalsize, anchor=south] at (0,0,0.6) {$r_3$};

\end{tikzpicture}
\caption{Weighted simplex $R_3 \subset \mathbb{R}^3_{\ge 0}$ for $n=4$,
shown with explicit coordinate axes. The vertices are
$(0,0,0)$, $(1,0,0)$, $(0,\tfrac12,0)$, and $(0,0,\tfrac13)$,
and the slanted face corresponds to the constraint
$r_1+2r_2+3r_3=1$.}
\label{fig:R3}
\end{figure}
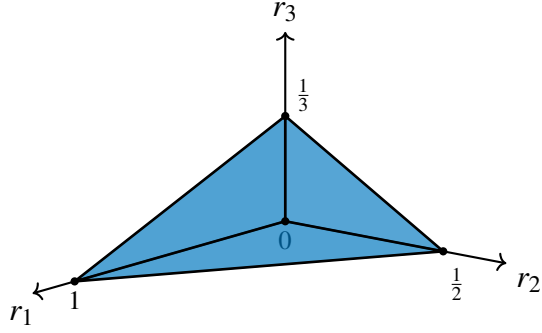

\subsection{Lie-algebraic structure: fundamental coweights}
\label{sec:Lie}

The Jacobian matrix $M_{ka}$ in~\eqref{eq:p_from_r_compact} has a
canonical Lie-algebraic interpretation  \cite{humphreys1972}.

\begin{prop}[Fundamental coweights of $\mathfrak{sl}(n)$ as spectral basis]
\label{prop:fundcoweights}
Define, for $a=1,\ldots,n-1$, the diagonal traceless $n\times n$ matrices: 
\begin{equation}
\label{eq:fund_coweight}
\omega_a^\vee
:= \frac{1}{n}
\diag\!\bigl(\underbrace{n-a,\ldots,n-a}_{a},
            \underbrace{-a,\ldots,-a}_{n-a}\bigr).
\end{equation}
They obey the following properties.
\begin{enumerate}
\item \textbf{Basis of the Cartan subalgebra.}
$\{\omega_1^\vee,\ldots,\omega_{n-1}^\vee\}$ is a basis of
\[
\mathfrak{h}
=\{\diag(d_1,\ldots,d_n):\textstyle\sum_k d_k=0\}
\subset\mathfrak{sl}(n).
\]

\item \textbf{Fundamental coweights.}
The simple roots $\alpha_j$, defined by
$\alpha_j(\diag(d_1,\ldots,d_n)):=d_j-d_{j+1}$, satisfy
\begin{equation}
\label{eq:coweight_def}
\alpha_j(\omega_a^\vee) = \delta_{ja},
\qquad 1\le j,a\le n-1.
\end{equation}
Hence the $\omega_a^\vee$ are precisely the \emph{fundamental
coweights} of $\mathfrak{sl}(n)$.

\item \textbf{Adjoint action on root vectors.}
\begin{equation}
\label{eq:adj_coweight}
[\omega_a^\vee,\,E_{j,j+1}] = \delta_{aj}\,E_{j,j+1}, \qquad 1\le j\le n-1, \qquad \left(E_{ij}\right)_{kl}=\delta_{ik}\,\delta_{jl}.
\end{equation}
\item \textbf{Change of basis to simple coroots.}
With $H_j:=E_{jj}-E_{j+1,j+1}$ the simple-coroot basis,
\begin{align}
\label{eq:omega_from_H}
\omega_a^\vee &= \sum_{j=1}^{n-1}
  \frac{\min(a,j)\bigl(n-\max(a,j)\bigr)}{n}\,H_j,\\
\label{eq:H_from_omega}
H_j &= \sum_{a=1}^{n-1}C_{ja}\,\omega_a^\vee,
\end{align}
where $C_{ja}=2\delta_{ja}-\delta_{j,a+1}-\delta_{j,a-1}$ is the
Cartan matrix of type $A_{n-1}$.

\item \textbf{Jacobian identity.}
The $(k,a)$ entry of the Jacobian satisfies
\begin{equation}
\label{eq:Jacobian_coweight}
\frac{\partial p_k}{\partial r_a} = M_{ka} = (\omega_a^\vee)_{kk}.
\end{equation}

\item \textbf{Compact real form.}
$\{\mathrm{i}\,\omega_a^\vee : 1\le a\le n-1\}$ is a Cartan subalgebra
basis of $\mathfrak{su}(n)$.
\end{enumerate}
\end{prop}

\begin{proof}
Items~(1) and~(6) are immediate.
For~(2): $\alpha_j(\omega_a^\vee)
=(\omega_a^\vee)_{jj}-(\omega_a^\vee)_{j+1,j+1}$.
If $j<a$: $(n-a)/n-(n-a)/n=0$.
If $j=a$: $(n-a)/n-(-a/n)=1$.
If $j>a$: $-a/n-(-a/n)=0$.
For~(3): $[D,E_{j,j+1}]=\alpha_j(D)\,E_{j,j+1}$
for any diagonal $D$, so~\eqref{eq:adj_coweight} follows from~(2).
For~(4): one verifies directly that
$(\omega_a^\vee)_{kk}=\sum_j(C^{-1}_{A_{n-1}})_{aj}(H_j)_{kk}$
using $(C^{-1}_{A_{n-1}})_{aj}=\min(a,j)(n-\max(a,j))/n$,
and inverts using $C$.
Item~(5) is a restatement of~\eqref{eq:p_from_r_compact}.
\end{proof}

\begin{remark}
The identity~\eqref{eq:Jacobian_coweight} expresses the key
structural fact: the spectral gap coordinates $r_a$ are
canonically dual to the simple roots $\alpha_a$ via the
fundamental coweights, explaining why all subsequent formulas
take their simplest form in $\br$-coordinates.
\end{remark}
\begin{example}
We give the expressions of the
coweights for $n=3$, i.e., case $A_2$: 
  \begin{equation}
\label{om3}
\omega_1^\vee= \frac{1}{3}\begin{pmatrix}
    2  & 0 & 0   \\
     0 &  -1 & 0\\
     0& 0 & -1
\end{pmatrix}, \quad \omega_2^\vee =  \frac{1}{3}\begin{pmatrix}
  1    & 0 & 0   \\
   0   &  1 & 0\\
   0 & 0 & -2
\end{pmatrix}.
\end{equation} 
\end{example}
\subsection{Density matrix and spectral decomposition}

From~\eqref{eq:p_from_r_compact} and Proposition~\ref{prop:fundcoweights}(5),
the traceless diagonal matrix $\sfD(\br)=\diag(p_1-1/n,\ldots,p_n-1/n)$
expands cleanly in the fundamental coweight basis:
\begin{equation}
\label{eq:D_r}
\sfD(\br) = \sum_{a=1}^{n-1} r_a\,\omega_a^\vee.
\end{equation}
The density matrix is therefore
\begin{equation}
\label{eq:rho_r}
\rho_{\br,\vph}
= \frac{1}{n}\,\mathbbm{1}_n
+ \sum_{a=1}^{n-1} r_a\, U(\vph)\,\omega_a^\vee\, U^\dagger(\vph)\equiv U(\vph)\,\rho_{\br}\, U^\dagger(\vph),
\end{equation}
where $U(\vph)\in\mathrm{SU}(n)$ (actually $\in \mathrm{SU}(n)/\mathbb{T}^{n-1}$),  parametrizes the eigenbasis.

\section{Statistical comparisons between \texorpdfstring{$R_{n-1}$}{R(n-1)} and \texorpdfstring{$\Delta_{n-1}^{\downarrow}$}{Delta(n-1)}}
\label{statcomp}

\subsection{Entropy and large deviations}

The volume ratio~\eqref{compVolR} is just  $n$, thus
$\frac{1}{n}\log\frac{\Vol(R_{n-1})}{\Vol(\Delta_{n-1}^{\downarrow})}\to 0$.
In large-deviation theory, two families of sets $A_n$, $B_n$ are
\emph{exponentially equivalent} if $\frac{1}{n}\log\frac{\Vol(A_n)}{\Vol(B_n)}\to0$.
Thus $R_{n-1}$ and $\Delta_{n-1}^{\downarrow}$ are exponentially equivalent:
they define the same large-deviation rate at speed $n$.

The passage from $\Delta_{n-1}$ to the ordered simplex $\Delta_{n-1}^{\downarrow}$
removes the combinatorial degeneracy under permutations, reducing the volume by a
factor $n!$ (entropy $\log n!$).  By contrast, the enlargement from
$\Delta_{n-1}^{\downarrow}$ to $R_{n-1}$ restores only the factor $n$ (entropy
$\log n$), negligible compared with $\log n!\sim n\log n$.  Both regions therefore
encode the same leading information content, differing only at logarithmic order.

Near the maximally mixed state, $p_i=\frac1n+\delta_i$ with $\sum_i\delta_i=0$,
the Shannon entropy satisfies
\[
H(p)=\log n-\frac{n}{2}\sum_i\delta_i^2+O(\|\delta\|^3),
\]
so entropy is controlled by quadratic fluctuations.  In $\br$-coordinates these
fluctuations are expressed naturally in terms of spectral gaps, which are precisely
the relevant degrees of freedom for decoherence, majorization, and monotone metrics.

\subsection{Metric considerations}

\subsubsection{Fisher--Rao metric in \texorpdfstring{$\br$}{r}-coordinates}

The Fisher–Rao metric~\cite{fisher1925,rao1945,amari2000} on the probability simplex
\begin{equation}
\label{eq:fisher-p}
g^{\mathrm{F}}=\sum_{k=1}^n\frac{\ud p_k^2}{p_k},\quad\sum_k\ud p_k=0,
\end{equation}
(up to the conventional factor $1/4$ used in some references), quantifies the statistical distinguishability of nearby spectra and provides the local quadratic approximation of the relative entropy. It 
pulls back to $R_{n-1}$ via~\eqref{eq:p_from_r_compact} as
\begin{equation}
\label{eq:fisher-r}
g^{\mathrm{F}}_{ab}(\br)
=\sum_{k=1}^n\frac{M_{ka}\,M_{kb}}{p_k(\br)},
\qquad a,b=1,\ldots,n-1.
\end{equation}
At the maximally mixed state  we have $\br=\mathbf{0}$, $p_k=1/n$ for all $k$, so
$g^{\mathrm{F}}_{ab}(0)=n\sum_k M_{ka}M_{kb}$.
Computing the sum using~\eqref{eq:p_from_r_compact} gives
\begin{equation}
\label{eq:fisher-r-origin}
g^{\mathrm{F}}_{ab}(0)
= \min(a,b)\bigl(n-\max(a,b)\bigr)
= n\,\bigl(C^{-1}_{A_{n-1}}\bigr)_{ab}.
\end{equation}
Thus, at the maximally mixed state, the Fisher--Rao metric in
gap coordinates is $n$ times the inverse Cartan matrix of
type $A_{n-1}$.  This is a direct consequence of
identity~\eqref{eq:Jacobian_coweight}: since $M_{ka}=(\omega_a^\vee)_{kk}$,
$$
g^{\mathrm{F}}_{ab}(0)
= n\,\mathrm{Tr}(\omega_a^\vee\omega_b^\vee)
= n\,(C^{-1}_{A_{n-1}})_{ab}.
$$
It is striking that a purely statistical object -- distinguishability of nearby
spectra -- collapses at the point of maximal symmetry into a purely
Lie-algebraic one: the gap coordinates diagonalize information geometry in
exactly the basis singled out by the root system of $\mathfrak{su}(n)$.

\subsubsection{Relative entropy near the maximally mixed state}

The relative entropy $D(p(\br)\|u_n)=\sum_k p_k\ln(np_k)$, $u_n =1/n$,
vanishes at $\br=\mathbf{0}$ and has zero gradient there
(since $\sum_k M_{ka}=0$).  Its Hessian equals
$g^{\mathrm{F}}_{ab}(0)$, giving the quadratic approximation
\begin{equation}
\label{eq:KL-r}
D(p(\br)\|u_n)
= \frac{n}{2}\sum_{a,b=1}^{n-1}\bigl(C^{-1}_{A_{n-1}}\bigr)_{ab}
  r_a\,r_b + O(\|\br\|^3).
\end{equation}
In short, Eq.~\eqref{eq:KL-r} states that relative entropy is,
to leading order, just the Fisher--Rao distance squared: the
same $A_{n-1}$ Cartan structure that fixes the metric also fixes
the entropic cost of opening each gap.

\subsubsection{Bures metric in \texorpdfstring{$\br$}{r}-coordinates}

For $\rho=U\,\diag(p_1,\ldots,p_n)\,U^\dagger$, the Bures metric
splits as
\begin{equation}
\label{eq:bures-split}
\ud s_{\mathrm{B}}^2
=\frac{1}{4}\sum_{k=1}^n\frac{\ud p_k^2}{p_k}
+\frac{1}{2}\sum_{i<j}
  \frac{(p_i-p_j)^2}{p_i+p_j}\,|\theta_{ij}|^2,
\qquad\theta:=U^\dagger\ud U.
\end{equation}
A key advantage of gap coordinates is the \emph{uniform} telescoping
formula: for all $1\le i<j\le n$,
\begin{equation}
\label{eq:gap-r}
p_i - p_j = \sum_{a=i}^{j-1}r_a,
\end{equation}
with no special case for pairs involving $p_n$.
Substituting~\eqref{eq:fisher-r} and~\eqref{eq:gap-r}
into~\eqref{eq:bures-split}:
\begin{equation}
\label{eq:bures-r}
\ud s_{\mathrm{B}}^2
=\frac{1}{4}\sum_{a,b}g^{\mathrm{F}}_{ab}(\br)\,\ud r_a\,\ud r_b
+\frac{1}{2}\sum_{1\le i<j\le n}
  \frac{\bigl(\sum_{a=i}^{j-1}r_a\bigr)^2}
       {p_i(\br)+p_j(\br)}\,|\theta_{ij}|^2.
\end{equation}
The angular cost of each mode $(i,j)$ is controlled by the
cumulative gap $\sum_{a=i}^{j-1}r_a$, which is linear in $\br$.
Equation~\eqref{eq:bures-r} makes this explicit: it splits the Bures metric
cleanly into a radial piece, carrying the full $A_{n-1}$ Cartan structure, and
an angular piece whose weights are simple linear combinations of the gaps,
with no separate treatment needed for pairs involving $p_n$.

\section{Advantages of the gap-coordinate parametrization}
\label{discint}

Sections~\ref{sec:statespace} and~\ref{statcomp} established, in turn, that the
gap coordinates $\mathbf r$ collapse the $n-1$ ordering constraints
$p_1\ge\cdots\ge p_n\ge0$ into the single linear inequality
$\sum_a a\,r_a\le1$, that they carry an intrinsic Lie-algebraic meaning as
fundamental-coweight coordinates of $\mathfrak{sl}(n)$
(Proposition~\ref{prop:fundcoweights}), and that every spectral gap
$p_i-p_j$ admits the uniform telescoping form~\eqref{eq:gap-r}, with no
special case for pairs involving $p_n$.

\medskip
\noindent
It is the combination of these three properties -- not any one in
isolation -- that makes $\mathbf r$ a more convenient spectral coordinate
than the raw eigenvalues $p_k$ for the metric, purity, and dynamical
computations of the sections that follow, in the same way that simple
roots, rather than arbitrary positive roots, provide the canonical basis
for root-system calculations in Lie theory. A precise comparison with the
eigenvalue coordinates used in the existing literature, and the resulting
novelty claim, is deferred to Section~\ref{sec:related}.
\section{A linear alternative to purity}
\label{sec:purity}

Let $\rho$ be an $n$-level quantum state with ordered eigenvalues
$p_1\ge\cdots\ge p_n\ge0$ and gap coordinates $r_a=p_a-p_{a+1}\ge0$.
We introduce the quantity
\begin{equation}
\label{eq:linear_radius}
\mathcal{R}(\rho)
:=\frac{n}{2(n-1)}\left\|\rho-\frac{\mathbbm{1}_n}{n}\right\|_1
=\frac{n}{2(n-1)}\sum_{i=1}^n\left|p_i-\frac{1}{n}\right|,
\end{equation}
where $\|\cdot\|_1$ denotes the trace norm.

We now give the expression of \eqref{eq:linear_radius} in gap coordinates.
Since $r_a\ge0$ and eigenvalues are ordered, the deviations
$p_k-1/n$ change sign exactly once: define the \emph{crossover index}
\begin{equation}
\label{eq:crossover}
k^*(\br):=\max\!\left\{k:p_k(\br)\ge\frac{1}{n}\right\},
\end{equation}
so that $p_k\ge1/n$ for $k\le k^*$ and $p_k<1/n$ for $k>k^*$.  The cases in which one or more $p_{k^*} = 1/n$  introduce a degeneracy in the above definition; however, as they form a set of measure zero, they are excluded from the present analysis. 

 Using the Jacobian identity $M_{ka}=(\omega_a^\vee)_{kk}$ and
summing over $k=1,\ldots,k^*$,
\begin{equation}
\label{eq:R_r}
\mathcal{R}(\rho)
=\frac{n}{n-1}
\sum_{a=1}^{n-1}
\bigl(C^{-1}_{A_{n-1}}\bigr)_{a,k^*}\,r_a
=\frac{n}{n-1}
\sum_{a=1}^{n-1}
\frac{\min(a,k^*)\bigl(n-\max(a,k^*)\bigr)}{n}\,r_a,
\end{equation}
where $C^{-1}_{A_{n-1}}$ is the inverse Cartan matrix of type $A_{n-1}$.
On each region $\{k^*=m\}$ of the weighted simplex $R_{n-1}$,
$\mathcal{R}(\rho)$ is \emph{linear} in $\br$; the full function is
piecewise linear with $n-1$ pieces corresponding to $m=1,\ldots,n-1$.

For lower cases for $n$ we have:
\begin{itemize}
\item $n=2$: $k^*=1$ always, $C^{-1}_{A_1}=(1/2)$, and
  $\mathcal{R}(\rho)=r_1=p_1-p_2$, the Bloch-sphere radius. 
\item $n=3$, $k^*=1$ (i.e.\ $r_2<r_1$):
  $\mathcal{R}(\rho)=r_1+\tfrac{1}{2}r_2$.
\item $n=3$, $k^*=2$ (i.e.\ $r_2\ge r_1$):
  $\mathcal{R}(\rho)=\tfrac{1}{2}r_1+r_2$.
\end{itemize}

The functional $\mathcal{R}(\rho)$ is unitarily invariant and satisfies:
\begin{itemize}
\item $\mathcal{R}(\rho)=0$ if and only if $\rho=\mathbbm{1}_n/n$
  (all $r_a=0$);
\item $\mathcal{R}(\rho)=1$ for any pure state;
\item $\mathcal{R}(\rho)\in[0,1]$ on all of $R_{n-1}$;
\item non-differentiability occurs only on the measure-zero
  hypersurfaces $\{k^*(\bs)=m\}$ where consecutive eigenvalues
  cross $1/n$, i.e.\ $p_m=1/n$.
\end{itemize}

Let us  discuss the interest  of the quantity $\mathcal{R}(\rho)$. 

In terms of advantages we can assert the following points. 

\textit{(i) Natural generalization of the qubit radius.}
For $n=2$, $\mathcal{R}(\rho)=r_1=p_1-p_2$ is the Bloch-sphere
radius.  Equation~\eqref{eq:R_r} is its canonical $n$-level
extension, replacing the single gap $r_1$ by a weighted combination
of all gaps, with weights given by the inverse Cartan matrix
evaluated at the crossover column $k^*$.

\medskip
\textit{(ii) Piecewise linearity and Cartan-matrix weights.}
On each linear piece $\{k^*=m\}$, equation~\eqref{eq:R_r} shows
that $\mathcal{R}$ weights the gap $r_a$ by
$(C^{-1}_{A_{n-1}})_{a,m}=\min(a,m)(n-\max(a,m))/n$.  Gaps near
the crossover index $m$ receive the largest weight; gaps far from
the crossover contribute less.  This gives a precise spectral
meaning to $\mathcal{R}$: it measures the total deviation from
uniformity, with each gap weighted by its distance to the crossover.

\medskip
\textit{(iii) Operational metric meaning.}
$\mathcal{R}(\rho)$ is proportional to the trace distance from
$\rho$ to the maximally mixed state, admitting a direct operational
interpretation in terms of distinguishability by optimal measurement.

\medskip
\textit{(iv) Convexity and decoherence monotonicity.}
Since the trace norm is convex, $\mathcal{R}$ satisfies
\[
\mathcal{R}\!\left(\sum_k\lambda_k\rho_k\right)
\le\sum_k\lambda_k\mathcal{R}(\rho_k).
\]
Mixing with the maximally mixed state gives
$\mathcal{R}(\lambda\rho+(1-\lambda)\mathbbm{1}_n/n)
=\lambda\,\mathcal{R}(\rho)$, so decoherence towards the maximally
mixed state contracts $\mathcal{R}$ linearly.

\medskip
\textit{(v) Clean domain and efficient parametrization.}
The weighted simplex $R_{n-1}=\{\br\ge0,\,\sum_a a\,r_a\le1\}$
parametrizes the ordered eigenvalue spectrum with a single linear
constraint, and represents $1/n!$ of the full eigenvalue simplex
$\Delta_{n-1}$.  This makes sampling and visualization of purity
landscapes significantly more efficient as $n$ grows.

We are naturally aware of the limitations outlined below.

\textit{(i) Non-smoothness.}
Unlike the quadratic purity $\gamma(\rho)=\mathrm{Tr}\,\rho^2$,
$\mathcal{R}(\rho)$ is not differentiable on the hypersurfaces
$\{p_m=1/n\}$.  This excludes it from gradient-based variational
methods without modification (e.g.\ smoothing by replacing
$|\cdot|$ by $\sqrt{(\cdot)^2+\varepsilon^2}$).

\medskip
\textit{(ii) Dependence on crossover index.}
The piecewise-linear structure means that the formula~\eqref{eq:R_r}
requires knowledge of $k^*(\br)$, which depends on the spectrum.
This makes symbolic manipulation more involved than for polynomial
invariants.

\medskip
\textit{(iii) Coarser spectral resolution.}
Different spectra with the same total $\ell^1$-deviation from $1/n$
have identical $\mathcal{R}$ but generically different purities
$\mathrm{Tr}\,\rho^2$.  $\mathcal{R}$ is therefore a coarser
spectral invariant than $\mathrm{Tr}\,\rho^2$.

\medskip
\textit{(iv) No polynomial structure.}
$\mathrm{Tr}\,\rho^2$ arises naturally in purity-decay rates under
Lindblad evolution and in the Hilbert--Schmidt geometry;
$\mathcal{R}(\rho)$ does not share this property.

In terms of originality and perspective, we think that
the two quantities are complementary probes of quantum mixedness.
The quantity $\mathcal{R}(\rho)$ is, up to normalization, the trace
distance from $\rho$ to the maximally mixed state $\mathbf{1}_n/n$,
an object that appears implicitly in various contexts in quantum
information theory~\cite{nielsen2000,helstrom1976,benatti2023,holevo2011}.  As a
purity measure normalized to $[0,1]$, it is distinct from the
dimensionally renormalized linear entropy
$\frac{n}{n-1}(1-\mathrm{Tr}\,\rho^2)$~\cite{BZ2006,breupetr02,holevo2011,watrous2018,benatti2023}, which is
polynomial in $\rho$ rather than spectral. Finally, what appears to be new is the explicit expression in gap
coordinates: on each linear piece $\{k^*=m\}$ of $R_{n-1}$,
equation~\eqref{eq:R_r} expresses $\mathcal{R}(\rho)$ as a weighted
sum of gaps $r_a$ with weights $(C^{-1}_{A_{n-1}})_{a,m}$, the
entries of the inverse Cartan matrix of type $A_{n-1}$.  This
connects a standard quantum-information quantity directly to the
root system of $\mathfrak{sl}(n)$, and does not appear to have
been noted in the literature.  

\section{\texorpdfstring{SU$(n)$}{SU(n)} parametrization}
\label{SUPar}

\subsection{General aspects}

From expression~\eqref{nrho}, only the left coset $\SU(n)/\mathbb{T}^{n-1}$ is relevant,
since elements of the maximal torus
\begin{equation}
\label{torus}
\mathbb{T}^{n-1}=\Big\{\diag\!\big(\eu^{i\theta_1},\ldots,\eu^{i\theta_n}\big)\in\SU(n)
:\sum_{j=1}^n\theta_j\equiv0\pmod{2\pi}\Big\}\cong\mathrm{U}(1)^{n-1}
\end{equation}
commute with the diagonal $\sfD(\br)$.  The torus has dimension $n-1$ (the rank of
$\SU(n)$), with Lie algebra $\mathfrak{t}=\{\ii\diag(\alpha_1,\ldots,\alpha_n):
\sum_j\alpha_j=0\}$; its Weyl group is $\mathrm{N}_{\SU(n)}(\mathbb{T}^{n-1})/\mathbb{T}^{n-1}\cong S_n$.

The total count checks out: $(n-1)+(n^2-n)=n^2-1$ real parameters for a generic
non-degenerate $\rho_{\br,\vph}$ (hermiticity and unit trace).

\subsection{Tilma--Sudarshan-like parametrization}

We now introduce  the embedded $\mathfrak{su}(2)$ generators for each pair $1\le i<j\le n$ in terms of the standard matrix units $(E_{ij})_{kl}=\delta_{ik}\delta_{jl}$:
\begin{equation}
\label{embedded-su2}
\sigma_1^{(i,j)}=\E{i}{j}+\E{j}{i},\quad
\sigma_2^{(i,j)}=-\ii(\E{i}{j}-\E{j}{i}),\quad
\sigma_3^{(i,j)}=\E{i}{i}-\E{j}{j}.
\end{equation}
For the diagonal sector we use the orthonormal Cartan generators
\begin{equation}
\label{cartan-generators}
\Hc{\ell}=\frac{1}{\sqrt{\ell(\ell+1)}}
\!\left(\sum_{j=1}^{\ell}\E{j}{j}-\ell\,\E{\ell+1}{\ell+1}\right),
\quad\ell=1,\ldots,n-1,\quad\mathrm{tr}(\Hc{\ell}\Hc{m})=\delta_{\ell m}.
\end{equation}

\noindent\textbf{Parametrization.}
Every $U_{\mathrm{full}}\in\SU(n)$ can be written, in lexicographic order, as
\begin{equation}
\label{paramSUn}
U_{\mathrm{full}}(\theta_{i,j},\phi_{i,j},\phi_\ell)
=\left(\prod_{j=2}^{n}\;\prod_{i=1}^{j-1}
R_{i,j}(\theta_{i,j},\phi_{i,j})\right)D(\phi_1,\ldots,\phi_{n-1}),
\end{equation}
where the elementary factors are embedded $\SU(2)$ rotations in the $(i,j)$-plane,
\begin{equation}
\label{Rij}
R_{i,j}(\theta,\phi)
=\exp\!\bigl(-\ii\tfrac{\phi}{2}\sigma^{(i,j)}_3\bigr)\,
\exp\!\bigl(-\ii\tfrac{\theta}{2}\sigma^{(i,j)}_2\bigr)\,
\exp\!\bigl(\ii\tfrac{\phi}{2}\sigma^{(i,j)}_3\bigr),
\end{equation}
and the diagonal Cartan-torus factor is
\begin{equation}
\label{Dphi}
D(\phi_1,\ldots,\phi_{n-1})=\exp\!\!\left(\ii\sum_{\ell=1}^{n-1}\phi_\ell\,\Hc{\ell}\right).
\end{equation}
For $n=2$ this recovers the qubit parametrization \eqref{Un2}.  The total parameter count is
$n(n-1)+( n-1)=n^2-1=\dim\SU(n)$.

\noindent\textbf{Notation.}
Modulo the right action of the Cartan torus, we write
\begin{equation}
\label{Ustnot}
U(\vph)=\prod_{i=1}^{n-1}\;\prod_{j=i+1}^{n}R_{i,j}(\theta_{i,j},\phi_{i,j})
=(\bu_1,\bu_2,\ldots,\bu_n),
\end{equation}
where $\bu_j=U(\vph)\be_j$ are the column vectors of $U$, forming an orthonormal
basis of $\C^n$.

\subsection{Example: \texorpdfstring{$n=3$}{n=3}}
\label{sec:n3}

For $n=3$ the coset element $U(\vph)\in\SU(3)/\mathbb{T}^2$ depends on six angles
$\theta_{1,2},\theta_{1,3},\theta_{2,3},\phi_{1,2},\phi_{1,3},\phi_{2,3}$, with
\[
U(\vph)=R_{1,2}(\theta_{1,2},\phi_{1,2})\,R_{1,3}(\theta_{1,3},\phi_{1,3})\,
R_{2,3}(\theta_{2,3},\phi_{2,3}).
\]
Setting $c_{ij}:=\cos\frac{\theta_{ij}}{2}$, $s_{ij}:=\sin\frac{\theta_{ij}}{2}$,
one obtains
\begin{align}
\label{Un3}
\nonumber&U(\vph)=\begin{pmatrix}
c_{12}c_{13} &
-\eu^{-i\phi_{12}}s_{12}c_{23}-\eu^{-i\phi_{13}+i\phi_{23}}c_{12}s_{13}s_{23} &
\eu^{-i(\phi_{12}+\phi_{23})}s_{12}s_{23}-\eu^{-i\phi_{13}}c_{12}s_{13}c_{23}\\[1ex]
\eu^{i\phi_{12}}s_{12}c_{13} &
c_{12}c_{23}-\eu^{i(\phi_{12}-\phi_{13}+\phi_{23})}s_{12}s_{13}s_{23} &
-\eu^{-i\phi_{23}}c_{12}s_{23}-\eu^{i(\phi_{12}-\phi_{13})}s_{12}s_{13}c_{23}\\[1ex]
\eu^{i\phi_{13}}s_{13} & \eu^{i\phi_{23}}c_{13}s_{23} & c_{13}c_{23}
\end{pmatrix}\\&\equiv(\bu_1,\bu_2,\bu_3).
\end{align}
The first column $\bu_1$ is a generic unit vector in $\C^3$ (a point of $\CP^2$);
the second column $\bu_2$ lies in $\bu_1^\perp$ and is obtained by the additional
rotation $R_{2,3}$; the third column $\bu_3$ is then uniquely fixed by
orthonormality, completing the flag manifold $\SU(3)/\mathbb{T}^2\simeq\mathcal{F}_3$ (see below).

\subsection{Coherent-state and flag-manifold structure}
\label{sec:param_flag}

\noindent\textbf{Perelomov coherent states.}
For the fundamental representation of $\SU(n)$ on $\C^n$, the first column
$\bu_1=U(\vph)\be_1$ is a Perelomov coherent state~\cite{perelomov86} with isotropy
subgroup $\mathrm{U}(n-1)$, so its orbit is $\SU(n)/\mathrm{U}(n-1)\simeq\CP^{n-1}$.
The rank-one projectors $|\bu_1\rg\lg\bu_1|$ then satisfy a resolution of the
identity over $\CP^{n-1}$ with respect to the Fubini--Study measure -- simply
the $i=1$ case of the general flag-manifold identity~\eqref{flag_resol} below.

\noindent\textbf{Flag manifold.}
Keeping the full ordered frame $U=(\bu_1,\ldots,\bu_n)$ modulo independent column
phases leads to the \emph{complete flag manifold}
$$
\mathcal{F}_n\simeq\SU(n)/\mathbb{T}^{n-1}= \left\{F\equiv[U] \in \SU(n)/\mathbb{T}^{n-1}\right\},
$$
where $[U]$ stands for the equivalence class of $U\in \mathrm{SU}(n)$ $\mathrm{mod}\,D\in \mathbb{T}^{n-1}$. 
Its nested structure is $\mathcal{F}_n\sim\CP^{n-1}\ltimes\CP^{n-2}\ltimes\cdots\ltimes\CP^1$:
at each step, $\bu_k$ defines a coherent state in the projective space of the
orthogonal complement of $\mathrm{span}\{\bu_1,\ldots,\bu_{k-1}\}$.  The normalized
invariant measure factorizes accordingly:
\begin{equation}
\label{flagmeas}
\ud\mu_{\mathcal{F}_n}
=\ud\mu_{\CP^{n-1}}(\bu_1)\,\ud\mu_{\CP^{n-2}}(\bu_2|\bu_1)\,
\cdots\,\ud\mu_{\CP^1}(\bu_{n-1}|\bu_1,\ldots,\bu_{n-2}).
\end{equation}
In the angular coordinates this reads explicitly
$$
\ud\mu_{\mathcal{F}_n}(F)
=\left(\prod_{m=1}^{n-1}\frac{m!}{(4\pi)^m}\right)
\prod_{1\le i<j\le n}\sin\theta_{i,j}\cos^{2(j-i-1)}\!\frac{\theta_{i,j}}{2}\,
\ud\theta_{i,j}\,\ud\phi_{i,j},
$$
with total volume 
\begin{equation}
\label{volFn}
\Vol(\mathcal{F}_n)=\dfrac{(4\pi)^{n(n-1)/2}}{\prod_{m=1}^{n-1}m!}
\end{equation}

\begin{prop}[Global flag resolution of the identity]
\label{prop:flag_resol}
For each $i=1,\ldots,n$, the column vectors $\bu_i(F)=U(\vph)\be_i$ satisfy
\begin{equation}
\label{flag_resol}
\int_{\mathcal{F}_n}n\,|\bu_i(F)\rg\lg\bu_i(F)|\,\ud\mu_{\mathcal{F}_n}(F)=\un_n.
\end{equation}
Consequently, the family $\{\rho_{\br,\vph}\}_{F\in\mathcal{F}_n}$ provides a
resolution of the identity:
\begin{equation}
\label{rho_resol}
\int_{\mathcal{F}_n}n\,\rho_{\br,\vph}\,\ud\mu_{\mathcal{F}_n}(F)=\un_n.
\end{equation}
\end{prop}
\begin{proof}
Define $A_i:=\int_{\mathcal{F}_n}|\bu_i\rg\lg\bu_i|\,\ud\mu_{\mathcal{F}_n}$.
By $\SU(n)$-invariance and Schur's lemma, $A_i=c_i\,\un_n$.  Taking the trace
and using $\int\ud\mu_{\mathcal{F}_n}=1$ gives $c_i=1/n$, proving~\eqref{flag_resol}.
Averaging~\eqref{nrho} over the flag manifold yields~\eqref{rho_resol} by linearity.
\end{proof}

\noindent\textbf{Covariant integral quantization.}
Proposition~\ref{prop:flag_resol}  provides, in the spirit of~\cite{befrigape22},
an $\SU(n)$-covariant integral quantization of functions on the flag manifold:
\begin{equation}
\label{quant}
f(F)\;\longmapsto\;\mathrm{Op}_f
:=\int_{\mathcal{F}_n}f(F)\,n\,\rho_{\br,\vph}\,\ud\mu_{\mathcal{F}_n}(F).
\end{equation}
This map sends classical observables on $\mathcal{F}_n\simeq\SU(n)/\mathbb{T}^{n-1}$
to operators on $\C^n$, and is a mixed-state generalization of Perelomov coherent
states where the complete flag manifold replaces the projective space.
It is of particular interest to identify the simplest symbols $f(F)$ whose
quantization $\mathrm{Op}_f$ reproduces hermitian elements of the Lie algebra
$\mathfrak{gl}(n,\C)$ itself -- most notably the fundamental
coweights~\eqref{eq:fund_coweight} -- ; partial results in this direction will
be presented in a forthcoming publication.

\noindent\textbf{Volume of the state space.}
The spectral--angular decomposition naturally endows the (non-degenerate) state
space $\mathcal{D}_n^\circ$ with the product measure $\ud^{n-1}\br\,\ud\mu_{\mathcal{F}_n}$.
Its total volume is simply the product of \eqref{Rn1Vol} and \eqref{volFn} : 
\begin{equation}
\label{volDncirc}
\Vol(\mathcal{D}_n^\circ)
=
\frac{(4\pi)^{n(n-1)/2}}{((n-1)!)^2\prod_{m=1}^{n-1}m!}.
\end{equation}
Degenerate spectra correspond to lower-dimensional strata (partial flag manifolds)
of measure zero, which do not contribute to this volume.

\section{GKLS dynamics in spectral--angular coordinates}
\label{sec:gkls}

\subsection{Framework}
We consider an $n$-level quantum system whose reduced dynamics is governed
by the GKLS equation
\begin{equation}
\label{GKLS}
\dot\rho=\mathcal{L}[\rho]=-\ii[H,\rho]+\mathcal{L}_{\mathrm{diss}}[\rho].
\end{equation}
Continuing the analysis of the open-qubit dynamical system defined by Eqs.  \eqref{dphi}, \eqref{dthet}, \eqref{dr}, we exploit the spectral--angular decomposition of $\rho_{\br,\vph}=U(\vph)\,\rho_{\br}\,U(\vph)^\dagger$ to separate  in order to separate:
\begin{itemize}
  \item the \emph{spectral variables} $\br=(r_1,\ldots,r_{n-1})$, encoding the
        eigenvalue gaps via~\eqref{def:r},
  \item the \emph{angular variables} $\vph$, encoding the eigenvectors on
        $\mathcal{F}_n$.
\end{itemize}

\subsection{Evolution equations}

\noindent\textbf{Time derivative of $\rho_{\br,\vph}$.}\\
Differentiating $\rho_{\br,\vph}=U(\vph)\rho_{\br}U^\dagger(\vph)$ and setting
$\Omega(\vph):=\dot{U}(\vph)U^\dagger(\vph)\in\mathfrak{su}(n)$, we obtain
\begin{equation}
\label{rhodot_decomp}
\dot\rho_{\br,\vph}=U(\vph)\dot\rho_{\br}U^\dagger(\vph)+[\Omega(\vph),\rho_{\br,\vph}].
\end{equation}

\noindent\textbf{Projection onto spectral and angular parts.}\\
To keep notation concise, we drop the $\vph$
 dependence.
Conjugating the GKLS equation by $U^\dagger$ and defining
$\widetilde{H}:=U^\dagger HU$,
$\widetilde{\mathcal{L}}_{\mathrm{diss}}[\rho_{\br}]:=U^\dagger\mathcal{L}_{\mathrm{diss}}[\rho]U$,
$\widetilde{\Omega}:=U^\dagger\Omega U$, one obtains
\begin{equation}
\label{split_eigenbasis}
\dot\rho_{\br}+[\widetilde{\Omega},\rho_{\br}]
=-\ii[\widetilde{H},\rho_{\br}]+\widetilde{\mathcal{L}}_{\mathrm{diss}}[\rho_{\br}].
\end{equation}

\noindent\textbf{Spectral equations.}\\
Taking diagonal matrix elements of~\eqref{split_eigenbasis},
the Hamiltonian commutator contributes nothing
($([\widetilde{H},\rho_{\br}])_{ii}=0$), giving
\begin{equation}
\label{spectral_eq}
\dot{p}_i
=D_i[\rho_{\br}],
\qquad
D_i[\rho_{\br}]:=
\bigl(\widetilde{\mathcal{L}}_{\mathrm{diss}}[\rho_{\br}]\bigr)_{ii},
\qquad i=1,\ldots,n.
\end{equation}
The eigenvalues evolve under dissipation only, independently of the Hamiltonian.

\noindent\textbf{Angular equations.}\\
Taking off-diagonal elements of~\eqref{split_eigenbasis}, for $i\ne j$ and
non-degenerate spectrum ($p_i\ne p_j$),
\begin{equation}
\label{omega_eq}
\widetilde{\Omega}_{ij}
=-\ii\widetilde{H}_{ij}
+\frac{\bigl(\widetilde{\mathcal{L}}_{\mathrm{diss}}[\rho_{\br}]\bigr)_{ij}}{p_i-p_j}.
\end{equation}
The angular dynamics is driven by both the Hamiltonian and dissipative contributions.
The factors $(p_i-p_j)^{-1}$ reflect the singular behaviour of the
parametrisation near degenerate spectra, consistent with the stratified
structure of the state space.

\subsection{Spectral equations in \texorpdfstring{$\br$}{r}-coordinates}

Differentiating the gap relations~\eqref{def:r} directly and using the
spectral equations~\eqref{spectral_eq} yields
\begin{equation}
\label{rdot}
\dot{r}_a
= \dot{p}_a-\dot{p}_{a+1}
= D_a[\rho_{\br}]-D_{a+1}[\rho_{\br}],
\qquad a=1,\ldots,n-1.
\end{equation}
Thus the spectral sector reduces to a closed system of $n-1$ equations
depending only on the diagonal part of the dissipator in the instantaneous
eigenbasis.  Each gap $r_a$ is driven by the \emph{difference} of two
adjacent diagonal dissipator entries, a structure that reflects the
definition~\eqref{def:r} directly.

\noindent\textbf{Recovering eigenvalue rates.}\\
Differentiating the inverse relations~\eqref{eq:p_from_r_compact} gives
\begin{equation}
\label{pdot_from_rdot}
\dot{p}_k
= \sum_{a=1}^{n-1} M_{ka}\,\dot{r}_a,
\qquad k=1,\ldots,n,
\end{equation}
with the matrix $M_{ka}$ defined in~\eqref{eq:p_from_r_compact}.
Expanding explicitly,
\begin{equation}
\label{pdot_explicit}
\dot{p}_k
= \sum_{a=k}^{n-1}\!\Bigl(1-\tfrac{a}{n}\Bigr)\dot{r}_a
  -\sum_{a=1}^{k-1}\tfrac{a}{n}\,\dot{r}_a,
\qquad k=1,\ldots,n,
\end{equation}
with the convention that empty sums vanish.
One verifies directly that $\sum_{k=1}^n\dot{p}_k=0$, as required by
trace preservation.

\subsection{Partial decoupling: summary}
\begin{itemize}
  \item \emph{Spectral variables}: the gap rates $\dot{r}_a=D_a-D_{a+1}$
        are determined solely by the diagonal dissipator entries;
        the Hamiltonian plays no role.
  \item \emph{Angular variables}: $\widetilde{\Omega}_{ij}$ depends on
        both Hamiltonian and dissipative terms via~\eqref{omega_eq}.
\end{itemize}
This separation provides a geometrically natural description of GKLS dynamics
in terms of spectral flows on the simplex of ordered eigenvalues and coherent
rotations on the flag manifold $\mathcal{F}_n$.

\subsection{Illustration: the real qutrit sector}
\label{sec:real-qutrit}

We now illustrate the general spectral--angular decomposition with the
 real case $n=3$  with real symmetric density
matrices (see the recent \cite{Curado-etal25} for the simplest non-trivial case, namely the real qubit).  In that case the angular part is described by an orthogonal
matrix $U\in \mathrm{SO}(3)$, depending on three Euler angles.

\paragraph{Spectral part.}
For \(n=3\), the gap coordinates are
$$
r_1=p_1-p_2,\qquad r_2=p_2-p_3,
$$
with
$$
p_1(\br)=\frac13+\frac23\,r_1+\frac13\,r_2,\qquad
p_2(\br)=\frac13-\frac13\,r_1+\frac13\,r_2,\qquad
p_3(\br)=\frac13-\frac13\,r_1-\frac23\,r_2.
$$
The admissible domain is the triangle
$$
r_1\ge 0,\qquad r_2\ge 0,\qquad r_1+2r_2\le 1,
$$
as is shown on the figure \ref{fig:R2Delta2}, left.
Hence the diagonal spectral state is
\begin{equation}
\label{rho-r-real-qutrit}
\rho_{\br}
=
\diag(p_1(\br),p_2(\br),p_3(\br)).
\end{equation}

\paragraph{Angular part.}
We use the standard \(zyz\) Euler decomposition
\begin{equation}
\label{SO3-euler}
U(\alpha,\beta,\gamma)=R_z(\alpha)\,R_y(\beta)\,R_z(\gamma),
\qquad
\alpha,\gamma\in[0,2\pi),\quad \beta\in[0,\pi],
\end{equation}
with
$$
R_z(\varphi)=
\begin{pmatrix}
\cos\varphi&-\sin\varphi&0\\
\sin\varphi&\cos\varphi&0\\
0&0&1
\end{pmatrix},
\qquad
R_y(\beta)=
\begin{pmatrix}
\cos\beta&0&\sin\beta\\
0&1&0\\
-\sin\beta&0&\cos\beta
\end{pmatrix}.
$$
The real qutrit state is then written as
\begin{equation}
\label{real-qutrit-rho}
\rho_{r_1,r_2,\alpha,\beta,\gamma}
=
U(\alpha,\beta,\gamma)\,\rho_{\br}\,U(\alpha,\beta,\gamma)^T .
\end{equation}

\paragraph{Real-preserving GKLS dynamics.}
To preserve the real symmetric sector, the Hamiltonian part must be generated
by a purely imaginary Hermitian matrix, equivalently
$$
H=\ii A,\qquad A^T=-A,\qquad A\in\mathfrak{so}(3).
$$
The GKLS equation may therefore be written as
\begin{equation}
\label{GKLS-real-qutrit}
\dot\rho=[A,\rho]+\mathcal{L}_{\mathrm{diss}}[\rho],
\end{equation}
where \(\mathcal{L}_{\mathrm{diss}}\) is assumed to preserve real symmetric
matrices.

Let
$$
\Omega:=U^T\dot U\in\mathfrak{so}(3),
\qquad
\widetilde{A}:=U^TAU,
\qquad
\widetilde{\mathcal{L}}_{\mathrm{diss}}[\rho_{\br}]
:=U^T\mathcal{L}_{\mathrm{diss}}[\rho]U .
$$
Then the conjugated equation reads
\begin{equation}
\label{split-real-qutrit}
\dot\rho_{\br}+[\Omega,\rho_{\br}]
=
[\widetilde A,\rho_{\br}]
+
\widetilde{\mathcal{L}}_{\mathrm{diss}}[\rho_{\br}] .
\end{equation}

\paragraph{Spectral equations.}
Writing
$$
d_i:=\bigl(\widetilde{\mathcal{L}}_{\mathrm{diss}}[\rho_{\br}]\bigr)_{ii},
\qquad i=1,2,3,
$$
the diagonal part of~\eqref{split-real-qutrit} gives
\begin{equation}
\label{pdot-real-qutrit}
\dot p_1=d_1,\qquad
\dot p_2=d_2,\qquad
\dot p_3=d_3,
\qquad
d_1+d_2+d_3=0.
\end{equation}
Hence the gap variables satisfy
\begin{equation}
\label{rdot-real-qutrit}
\dot r_1=d_1-d_2,\qquad
\dot r_2=d_2-d_3.
\end{equation}
As in the qubit case, the spectral variables are driven only by the
dissipative part.

\paragraph{Angular equations.}
Set
$$
k_{ij}:=
\bigl(\widetilde{\mathcal{L}}_{\mathrm{diss}}[\rho_{\br}]\bigr)_{ij},
\qquad
1\le i<j\le 3.
$$
Since \(\rho_{\br}=\diag(p_1,p_2,p_3)\), the off-diagonal part of
\eqref{split-real-qutrit} gives
\begin{equation}
\label{omega-ij-real-qutrit}
\Omega_{ij}
=
\widetilde A_{ij}
-
\frac{k_{ij}}{p_i-p_j},
\qquad i\neq j,
\end{equation}
that is,
\begin{equation}
\label{omega-gap-real-qutrit}
\Omega_{12}=\widetilde A_{12}-\frac{k_{12}}{r_1},
\qquad
\Omega_{23}=\widetilde A_{23}-\frac{k_{23}}{r_2},
\qquad
\Omega_{13}=\widetilde A_{13}-\frac{k_{13}}{r_1+r_2}.
\end{equation}

On the other hand, for the Euler decomposition~\eqref{SO3-euler}, one computes
\begin{equation}
\label{Omega-matrix-euler}
\Omega=U^T\dot U=
\begin{pmatrix}
0 & -(\dot\alpha\cos\beta+\dot\gamma)
  & \dot\alpha\sin\beta\sin\gamma+\dot\beta\cos\gamma \\[1mm]
\dot\alpha\cos\beta+\dot\gamma
  & 0
  & \dot\alpha\sin\beta\cos\gamma-\dot\beta\sin\gamma \\[1mm]
-\dot\alpha\sin\beta\sin\gamma-\dot\beta\cos\gamma
  & -\dot\alpha\sin\beta\cos\gamma+\dot\beta\sin\gamma
  & 0
\end{pmatrix}.
\end{equation}
Therefore,
\begin{align}
\label{alpha-dot-real-qutrit}
\dot\alpha
&=
\frac{\Omega_{13}\sin\gamma+\Omega_{23}\cos\gamma}{\sin\beta},
\\[1mm]
\label{beta-dot-real-qutrit}
\dot\beta
&=
\Omega_{13}\cos\gamma-\Omega_{23}\sin\gamma,
\\[1mm]
\label{gamma-dot-real-qutrit}
\dot\gamma
&=
-\Omega_{12}
-\cot\beta\,
\bigl(\Omega_{13}\sin\gamma+\Omega_{23}\cos\gamma\bigr).
\end{align}
Substituting~\eqref{omega-gap-real-qutrit} yields the explicit angular system
\begin{align}
\label{alpha-dot-full-real-qutrit}
\dot\alpha
&=
\frac{
\left(\widetilde A_{13}-\dfrac{k_{13}}{r_1+r_2}\right)\sin\gamma
+
\left(\widetilde A_{23}-\dfrac{k_{23}}{r_2}\right)\cos\gamma
}{\sin\beta},
\\[2mm]
\label{beta-dot-full-real-qutrit}
\dot\beta
&=
\left(\widetilde A_{13}-\frac{k_{13}}{r_1+r_2}\right)\cos\gamma
-
\left(\widetilde A_{23}-\frac{k_{23}}{r_2}\right)\sin\gamma,
\\[2mm]
\label{gamma-dot-full-real-qutrit}
\dot\gamma
&=
-\left(\widetilde A_{12}-\frac{k_{12}}{r_1}\right)
-\cot\beta
\left[
\left(\widetilde A_{13}-\frac{k_{13}}{r_1+r_2}\right)\sin\gamma
+
\left(\widetilde A_{23}-\frac{k_{23}}{r_2}\right)\cos\gamma
\right].
\end{align}

\paragraph{Interpretation.}
Equations~\eqref{rdot-real-qutrit} and
\eqref{alpha-dot-full-real-qutrit}--\eqref{gamma-dot-full-real-qutrit}
constitute the exact analogue, for real qutrits, of the qubit equations
\eqref{dphi}--\eqref{dr} given in the introduction.  The two spectral variables
$r_1,r_2$ evolve only through the dissipative diagonal terms $d_i$, while
the Euler angles $\alpha,\beta,\gamma$ are driven by both the Hamiltonian
part ($\widetilde A_{ij}$) and the dissipative off-diagonal terms
($k_{ij}$).  The singular denominators $r_1$, $r_2$, and $r_1+r_2$
signal the expected breakdown of angular coordinates at spectral degeneracies,
namely $p_1=p_2$, $p_2=p_3$, or $p_1=p_3$.  An interesting special case might arise when $k_{ij} = f_{ij}(\bm{\vph})\,(p_i - p_j)$,
i.e.\ explicitly
\begin{equation}
    k_{12} = f_{12}(\bm{\vph})\,r_1, \qquad
    k_{13} = f_{13}(\bm{\vph})\,(r_1+r_2), \qquad
    k_{23} = f_{23}(\bm{\vph})\,r_2,
\end{equation}
for some purely angular functions $f_{ij}(\bm{\vph})$. In this case the
ratios $k_{ij}/(p_i - p_j) = f_{ij}(\bm{\vph})$ depend only on $\bm{\vph}$,
and the angular evolution system becomes autonomous, i.e.\ completely decoupled
from~$\br$. Note that this condition is not generically satisfied for an
arbitrary Lindblad model, but holds for specific classes of jump operators.
 A natural physical setting in which this factorization holds is the
\emph{secular} (or \emph{Born--Markov}) approximation~\cite{breupetr02,alilend07}:
in that framework, the Lindblad operators reduce to Bohr-frequency transition
operators in the energy eigenbasis, so that coherences decay multiplicatively
and the off-diagonal dissipator elements in the instantaneous eigenbasis
of~$\rho$ satisfy $k_{ij} = f_{ij}(\bm{\vph})\,(p_i-p_j)$, rendering the
angular equations~\eqref{alpha-dot-full-real-qutrit}--\eqref{gamma-dot-full-real-qutrit}
autonomous.

 This explicit real example with 
$n=3$ illustrates how the gap coordinates $(r_1,r_2)$
 emerge as the natural spectral variables governing the GKLS flow, while the Euler angles capture the orthogonal motion of the eigenframe. The case $n=3$ is not merely a convenient low-dimensional illustration: it corresponds precisely to the trichromatic structure of human colour vision, where the three cone-cell types define a natural $3\times 3$
 density-matrix framework. Interpreting the GKLS dissipative dynamics within this framework opens a novel route to modelling chromatic adaptation and colour discrimination - phenomena traditionally described by purely geometric or psychophysical methods - as the irreversible evolution of an open quantum-like system. A thorough analysis of this instructive case and its perceptual consequences is currently under development \cite{bachiri-etal26}.

\section{Related parametrizations and the precise novelty claim}
\label{sec:related}

We compare our construction with the four papers whose overlap with this
work is closest, in order of relevance.

\paragraph{Coset parametrization and its Bures metric~\cite{akhtarshenas07,akhtarshenas07bures}.}
Akhtarshenas parametrizes $\mathcal D_n$ by writing $U\in\mathrm U(n)/\mathbb
T^n$ as a product of $n-1$ coset factors $\Omega^{(m;n)}\in
\bigl(\mathrm U(m)\otimes\mathbb T^{n-m}\bigr)/\bigl(\mathrm
U(m-1)\otimes\mathbb T^{n-m+1}\bigr)$, and computes the H\"ubner formula for
the Bures metric explicitly in these coordinates for \emph{arbitrary}
$n$~\cite{akhtarshenas07bures}, obtaining a block-diagonal metric $g=g^{(D)}\oplus
g^{(C)}$ with $g^{(D)}$ the eigenvalue block and $g^{(C)}$ the coset block.
This is more general than what we do: his eigenvalue coordinates
$\theta_k$ (defined by $\lambda_k=\sin^2\theta_1\cdots\sin^2\theta_{k-1}\cos^2\theta_k$)
and his coset block are valid at an arbitrary point of $\mathcal D_n$, not
only near the maximally mixed state. What his construction does not provide
is any algebraic interpretation of the eigenvalue coordinates themselves: $\theta_k$
parametrizes the simplex but carries no Lie-theoretic meaning. Our
contribution is narrower in scope --- we work in gap coordinates and state
the metric explicitly only at $\rho=\mathbbm 1_n/n$ --- but it identifies
the eigenvalue coordinates with the fundamental coweights of $A_{n-1}$
(Proposition~\ref{prop:fundcoweights}) and shows that the resulting Fisher--Rao
and Bures metrics at that point are $n$ times the \emph{inverse Cartan
matrix} of $A_{n-1}$ (eq.~\eqref{eq:fisher-r-origin}), a structural fact
absent from~\cite{akhtarshenas07bures}.

\paragraph{Coset and Jarlskog parametrizations, review~\cite{bruning-etal11}.}
Br\"uning, M\"akel\"a, Messina and Petruccione review the Bloch-vector,
coset and Jarlskog parametrizations of $\mathcal D_n$ and note explicitly
that the coset approach, as such, ``offers no parametrization of individual
density matrices'' for general $n$; the Jarlskog construction is the one they
present as achieving injectivity, via a recursive factorization
$U_n=A_{n,n}\cdots A_{n,2}$ of $\mathrm{SU}(n)$~\cite{jarlskog05}. Both of
their eigenvalue sectors are the raw ordered simplex $\Lambda_n=\{\lambda_1\ge
\cdots\ge\lambda_n\ge0,\ \sum\lambda_j=1\}$, subject to $n-1$ ordering
inequalities. Nothing in that review addresses the Bures/Fisher--Rao metric,
a purity functional, or GKLS dynamics; its overlap with the present paper is
limited to the shared shape~\eqref{eq:generic-split}. Our gap coordinates
replace the $n-1$ ordering inequalities by the single inequality
$\sum_a a\,r_a\le1$ (Section~\ref{sec:statespace}), which is the one concrete
improvement we claim over both the coset and Jarlskog eigenvalue sectors.

\paragraph{Flag-based control of Lindblad orbit dynamics~\cite{rooney-bloch-rangan16}.}
This is the closest overlap in substance. Rooney, Bloch and Rangan write the
Lindblad equation in terms of a possibly-degenerate flag $F=(P_1,\ldots,P_{n_d})$
of eigenprojectors and an eigenvalue vector $\Lambda$ obeying the linear ODE
$\dot\Lambda=\Omega^\pi\Lambda$ (their eq.~(17)), and show that the flag
motion --- \emph{intra-orbit} dynamics --- is driven by the full Hamiltonian
control, while the eigenvalues --- \emph{inter-orbit} motion --- evolve
independently once fast unitary control is assumed; the qualitative
statement that eigenvalues are Hamiltonian-blind is essentially theirs. We
do not claim priority on this decoupling. What we add is: (i) an explicit,
non-degenerate coordinate realization of both sectors --- gap coordinates
$\mathbf r$ for the spectrum and Tilma--Sudarshan angles $\vph$ for the flag
--- which their projector calculus deliberately avoids because it must
handle degenerate flags and control-theoretic reachability questions; (ii)
the direct re-derivation of the decoupling from the GKLS equation in these
coordinates (eqs.~\eqref{spectral_eq}--\eqref{omega_eq}), stated without
any control-theoretic apparatus (no controllability, no reachable sets,
none of which we address); and (iii) the metric and purity structure of
Sections~\ref{statcomp}--\ref{sec:purity}, which has no counterpart in their
paper.

\paragraph{Trees, forests and stationary states of Lindblad systems~\cite{rooney-bloch-rangan18}.}
This paper shares only the common ancestor equation
$\dot\Lambda=\Omega^\pi\Lambda$ with~\cite{rooney-bloch-rangan16} (and, more
distantly, with Hioe and Eberly's coherence-vector
equations~\cite{hioe-eberly81}). Its actual content --- a combinatorial
characterization, via trees and forests on a weighted digraph, of the
\emph{kernel} of $\Omega$ and hence of the asymptotic end-states of
generalized-permutation-matrix Lindblad systems --- addresses long-time
stationary behavior, a question we do not touch. 
We mention this connection solely because of the shared starting equation;
no further methodological or result-level overlap exists with the present paper.

\bigskip
\noindent In summary: the spectral/angular splitting~\eqref{eq:generic-split}
and its GKLS decoupling are established results, to which we do not add a
new qualitative claim. What is new, and is the sole basis of our novelty
claim, is (a) the identification of gap coordinates with the fundamental
coweights of $A_{n-1}$ (Proposition~\ref{prop:fundcoweights}); (b) the
resulting Cartan-matrix form of the Fisher--Rao and Bures metrics at the
maximally mixed state (Section~\ref{statcomp}); (c) the piecewise-linear
purity functional $\mathcal R(\rho)$ in gap coordinates
(Section~\ref{sec:purity}); and (d) the explicit Tilma--Sudarshan
coordinatization of the flag manifold together with its coherent-state
resolution of the identity and covariant quantization
(Section~\ref{SUPar}), none of which appear in
\cite{akhtarshenas07bures,bruning-etal11,rooney-bloch-rangan16,rooney-bloch-rangan18}.


\section{Conclusion}
\label{sec:conclusion}

We have introduced a spectral--angular parametrization of finite-dimensional density
matrices, $\rho_{\br,\vph}=U(\vph)\rho_{\br}U(\vph)^\dagger$, which explicitly separates
the spectral variables (encoded in the convex polytope $R_{n-1}$) from the angular
variables (living on the flag manifold $\mathcal{F}_n\simeq\SU(n)/\mathbb{T}^{n-1}$).
This provides a natural extension of the Bloch-ball description of qubits to arbitrary
dimension.

Several structural consequences follow.
\begin{itemize}
\item The state space factorizes into a product of a spectral polytope and a
  homogeneous space, making the role of unitary orbits explicit.
\item The spectral variables $\br$ provide a linear, ordered description of eigenvalues
  directly adapted to majorization, entropy variations, and information geometry.
  \item The GKLS dynamics exhibits a natural partial decoupling: the Hamiltonian part
  acts on angular variables while the dissipative part governs the spectral evolution.
  \item Any nondegenerate density matrix has the following expansion $\rho_{\br,\vph}= \frac{\un_n}{n} + \sum_a r_a U(\vph)\omega_a^\vee U^{\dag}(\vph)$  in terms of  the
fundamental coweight basis of $\mathfrak{sl}(n)$. This Lie algebraic feature  should be of efficient help in making explicit the dissipative part of the GKLS equations.
\item The induced volume factorizes into spectral and angular contributions, yielding an
  explicit formula for the volume of the generic state space.
\item The flag manifold provides a global coherent-state resolution of the identity~\eqref{rho_resol},
  enabling $\SU(n)$-covariant integral quantization~\eqref{quant}.
\end{itemize}

From an information-geometric viewpoint, the $\br$-coordinates appear as natural gap
variables controlling entropy production and large deviations around the maximally
mixed state, while entering directly the structure of monotone metrics (Fisher--Rao,
Bures). Among the byproducts of this parametrization, we single out the piecewise-linear
purity functional $\mathcal{R}(\rho)$, introduced in Section~\ref{sec:purity} as the
natural $n$-level extension of the Bloch-sphere radius: up to normalization, it is the
trace distance from $\rho$ to the maximally mixed state, and hence carries a direct
operational meaning in terms of state discrimination, while remaining exactly linear in
the gap coordinates $\br$ on each region $\{k^*=m\}$ of the weighted simplex $R_{n-1}$,
with weights given by the inverse Cartan matrix of $A_{n-1}$. Unlike the quadratic
purity $\mathrm{Tr}\,\rho^2$, $\mathcal{R}(\rho)$ contracts linearly under mixing with
the maximally mixed state and is computable without recourse to matrix powers, at the
price of a non-smoothness confined to the measure-zero crossover hypersurfaces
$\{p_m=1/n\}$. We regard $\mathcal{R}(\rho)$ and $\mathrm{Tr}\,\rho^2$ as complementary
probes of quantum mixedness -- one spectral and piecewise linear, the other polynomial
and smooth -- and its explicit expression~\eqref{eq:R_r} in terms of the inverse Cartan
matrix of $\mathfrak{sl}(n)$ appears to be new.

The precise relation of $\br$ to the eigenvalue coordinates used in the
existing literature, and the resulting novelty claim, was stated in
Section~\ref{sec:related}.

Beyond our ongoing study of colour perception \cite{bachiri-etal26}, several
directions invite further investigation. On the physical side, the
higher-dimensional cases $n>3$ offer a natural testing ground for spectral
relaxation and decoherence in multilevel open quantum systems, where the full
complexity of the gap-coordinate geometry becomes essential. On the mathematical
side, the covariant integral quantization on the flag manifold
$\mathcal{F}_n\simeq\SU(n)/\mathbb{T}^{n-1}$ deserves systematic development: the interplay
between the quantizer-dequantizer formalism and the stratified structure of the
flag variety remains largely unexplored. We will
develop this  full quantization framework, including the semi-classical phase-space
structures, in a forthcoming paper.
Finally, the polyhedral geometry of the
spectral domain $R_{n-1}$ - a convex polytope whose combinatorial structure
encodes permutation symmetry - connects naturally to majorization theory and
the theory of doubly stochastic maps, suggesting deeper links between the GKLS
dissipative flow and classical notions of disorder and entropy.

A more ambitious direction, which requires careful treatment, concerns the extension to the infinite-dimensional
setting, which acquires a precise algebraic meaning through the infinite simple
Lie algebra $A_{\infty} = \varinjlim \mathfrak{sl}(n)$, whose Dynkin diagram is
the semi-infinite chain $\circ - \circ - \circ - \cdots$. This perspective
sheds new light on the gap coordinates themselves: in the finite case
$A_{n-1} = \mathfrak{sl}(n)$, the simple roots $\alpha_i = e_i - e_{i+1}$ are
precisely the spectral gaps $r_i = p_i - p_{i+1}$, so the natural
spectral variables of the GKLS flow are nothing other than simple root
coordinates in the Cartan subalgebra of $A_{n-1}$. In the limit $A_{\infty}$,
corresponding to trace-class density operators on a separable Hilbert space
$\mathcal{H}$, the gap sequence $(r_1, r_2, \ldots)$ becomes an infinite
positive root string subject to the summability constraint $\sum_i r_i < \infty$
imposed by $\{\lambda\}=p_1,p_2,\ldots \in \ell^1$, while the finite flag manifold
$\SU(n)/\mathbb{T}^{n-1}$ gives way to the flag ind-variety
$\mathrm{SU}(\infty)/\mathbb{T}^{\infty}$ in the sense of Shafarevich-Dimitrov-Penkov,
which becomes the configuration space of the eigenframe. The GKLS flow in this
setting requires careful treatment of domain conditions and unbounded generators,
and the covariant integral quantization must be reformulated in this
ind-geometric setting. Nevertheless, the majorization partial order extends
naturally to $\ell^1$-summable eigenvalue sequences, and the representation
theory of $A_{\infty}$ - in particular its highest-weight modules and the
associated Fock-space constructions - may provide the appropriate
functional-analytic framework for a unified treatment of the GKLS dissipative
flow across all dimensions simultaneously, with potential connections to quantum
field theory and continuous-variable quantum information.


\begin{acknowledgments}
J.-P.~G.\ gratefully acknowledges Universit\'e Mohammed V (Rabat), Mohammed VI
Polytechnic University (UM6P) and Ibn Tofa\"{i}l university, Kenitra, for their kind invitations and financial support, which
provided an excellent environment for collaboration and significantly contributed to the
completion of this work.
\end{acknowledgments}

%


\end{document}